\documentclass[a4paper]{article}
\usepackage[english]{babel}

\usepackage{amsmath}
\usepackage{amssymb}
\usepackage[numbers]{natbib}
\usepackage{fullpage}
\usepackage{graphicx}
\usepackage{subfigure}
\usepackage{amsthm}
\usepackage{todonotes}
\usepackage[]{algorithm2e}
\usepackage{xspace}
\usepackage{pgf,pgfarrows,pgfnodes}
\usepackage{tikz}
\usepackage{subfigure}
\usepackage{tcolorbox}
\usepackage{soul}
\usepackage[colorlinks=true,allcolors=blue]{hyperref}
\usepackage{slashbox}
\usepackage{amssymb}
\usepackage{pifont}
\newcommand{\cmark}{\ding{51}}%
\newcommand{\xmark}{\ding{55}}%

\newtheorem{theorem}{Theorem}
\newtheorem{corollary}{Corollary}

\newtheorem{lemma}{Lemma}
\newtheorem{definition}{Definition}
\newtheorem{claim}{Claim}

\newcommand{\fixed}{\texttt{Fixed}\xspace}
\newcommand{\fixedp}{\ensuremath{\texttt{Fixed}^+}\xspace}
\newcommand{\fixedzo}{\ensuremath{\texttt{Fixed}^{\{0,1\}}}\xspace}
\newcommand{\fixedzm}{\ensuremath{\texttt{Fixed}^{\{0,-1\}}}\xspace}
\newcommand{\rand}{\texttt{Random}\xspace}
\newcommand{\randp}{\ensuremath{\texttt{Random}^+}\xspace}
\newcommand{\calth}{\ensuremath{\mathcal{T}_h}\xspace}

\newcommand{\caltl}{\ensuremath{\mathcal{T}_l}\xspace}
\newcommand{\reals}{\ensuremath{\mathbb{R}}\xspace}
\newcommand{\ocal}{\ensuremath{\mathcal{O}}\xspace}

\newcommand{\dist}{\ensuremath{\texttt{dist}}\xspace}

\newcommand{\ybf}{\ensuremath{\mathbf{y}}\xspace}
\newcommand{\ybfs}{\ensuremath{\mathbf{y}^*}\xspace}
\newcommand{\ybfp}{\ensuremath{\mathbf{y}'}\xspace}

\newcommand{\opt}{\textsc{OPT}\xspace}
\newcommand{\util}{\textsc{Egalitarian}\xspace}
\newcommand{\welfare}{\textsc{Utilitarian}\xspace}
\newcommand{\happy}{\textsc{Happiness}\xspace}

\newcommand{\yj}{\ensuremath{\ybf_j}\xspace}
\newcommand{\tij}{\ensuremath{t_{ij}}\xspace}
\newcommand{\ti}{\ensuremath{t_{i}}\xspace}

\newcounter{alg}
\newcommand{\alg}[1]{\refstepcounter{alg}\label{#1}}

\newcount\Comments  
\definecolor{darkgreen}{rgb}{0,0.6,0}
\newcommand{\kibitz}[2]{\ifnum\Comments=1{\color{#1}{#2}}\fi}
\newcommand{\lefteris}[1]{\kibitz{blue}{[LEFTERIS: #1]}}

\usetikzlibrary{calc}
\usepackage{zref-savepos}

\newcounter{NoTableEntry}
\renewcommand*{\theNoTableEntry}{NTE-\the\value{NoTableEntry}}

\renewcommand\hl[1]{#1} 

\begin{document}

\title{\vspace{-0.5cm} Heterogeneous Facility Location Games}
\author{\vspace{-1.0cm}Eleftherios Anastasiadis 
\thanks{Imperial College London. Email: 
\texttt{e.anastasiadis@imperial.ac.uk}}
\and Argyrios Deligkas
\thanks{Royal Holloway University of London. Email: 
\texttt{argyrios.deligkas@rhul.ac.uk}}}
\date{}
\maketitle

\begin{abstract}
We study heterogeneous $k$-facility location games. In this model, there are $k$ 
facilities where each facility serves a different purpose. 
Thus, the preferences of the agents over the facilities can vary arbitrarily. 
Our goal is to design strategy-proof mechanisms that place the facilities in a
way to maximize the minimum utility among the agents.
For $k=1$, if the agents' locations are known, we prove that the mechanism that 
places the facility on an optimal location is strategy-proof.
For $k \geq 2$, we prove that there is no optimal strategy-proof mechanism, 
deterministic or randomized, even when $k=2$ there are only two agents with 
known locations and the facilities have to be placed on a line segment. 
We derive inapproximability bounds for deterministic and randomized strategy 
proof mechanisms. 
Finally, we focus on the line segment and provide strategy-proof mechanisms that achieve 
constant approximation. All of our mechanisms are simple and communication efficient.
As a byproduct, we show that some of our mechanisms can be used to achieve constant 
factor approximations for other objectives as the social welfare and the happiness. 
\end{abstract}

\section{Introduction}

Facility location games lie in the intersection of AI, game theory, and social choice theory. The
basic version of the problem has been widely studied in the literature \cite{Mo80, barbera1994characterization, schummer2002strategy}. 
In this setting, a central planner has to locate a facility on a real line based 
on the \emph{reported} locations of selfish agents who want to be as close as possible
to the facility. The goal of the planner is to locate the facility in a way that the 
sum of the utilities of the agents is maximized.~\footnote{In~\cite{PT09} the objective 
was to minimize the social cost.} 
However, the agents can \emph{misreport} their locations in order to manipulate 
the planner and increase their utility. 
One main objective of the planner is to design procedures to locate the facility, 
called \emph{mechanisms}, that incentivize the agents to report their true 
locations, i.e., the mechanisms are \emph{strategy-proof}. 

When monetary payments are not allowed, that is the planner cannot pay the agents 
or demand payments from them; it is not always possible to design mechanisms 
that implement an optimal solution and remain strategy-proof. 
Thus, the goal is to design mechanisms that \emph{approximately} maximize an 
objective function under the constraint that they are strategy-proof. 
The term \emph{approximate mechanism design without money}, introduced 
by~\citeauthor{PT09}, is usually deployed for problems like the one described above.
\citeauthor{PT09} studied \emph{homogeneous} facility location games, where 
one, or two, \emph{identical} facilities had to be placed on a real line and every
agent wanted to be as close as possible to any of them. In this setting, the 
agents were reporting to the planner a point on the line and the objectives 
studied were the maximization of the \emph{social welfare} or the 
\emph{minimum utility} among the agents.

In many real-life scenarios, though, both facilities and the preferences of the agents 
are \emph{heterogeneous}; every facility serves a different need and every agent 
has potentially different needs from the others. 
Consider, for example, the case where the government is planning to build a school 
and a factory. Citizens' preferences for these facilities might significantly 
differentiate. Those who work at the factory and also have children that go to school
wish both facilities to be built close to their homes. Citizens without children might 
want the school to be built far because of the noise. Finally, those who do not 
work at the factory prefer its location to be far from their home to avoid  
emitted pollution.

The example above shows that an agent might want to be \emph{close} to a facility, 
be \emph{away} from a facility, or be \emph{indifferent} about its presence. 
\citeauthor{FJ15}~\cite{FJ15} studied 1-facility heterogeneous 
games where each agent reported his preferred location on the line, while it was 
known to the planner, whether he wanted to be close to, or away from, the facility.
\citeauthor{ZL15} \cite{ZL15} extended the model of~\cite{FJ15}
for heterogeneous 2-facility games and studied the social utility objective for 
several different scenarios of the information the planner knows.
%
%
\citeauthor{SV16}~\cite{SV16} studied heterogeneous 2-facility games
on discrete networks. In their setting, each agent is located on a node of a 
graph and either is indifferent or wants to be close to each facility and the
planner knows the location of every agent but not their preferences for the 
facilities.
 
In this paper, we extend the aforementioned models and study heterogeneous 
$k$-facility location games; simply $k$-facility games. 
Our main focus is to maximize the minimum utility among all the agents, termed
\util. As a byproduct, we derive results for the social welfare, termed
\welfare, and the recently proposed minimum \emph{happiness} objective,
termed \happy.
\happy, which is reminiscent of the proportionality notion in resource allocation problems,  
is a fairness criterion for facility location problems introduced in~\cite{MLYZ}. 
The happiness of an agent is the ratio between the utility he gets under the 
locations of the facilities over the maximum utility the agent could get under any location.
To the best of our knowledge, there is no prior work on this model. We note that 
while our model is a natural extension of the aforementioned models 
almost none of those results apply in our case.

\subsection{Our contributions}
We study several questions regarding heterogeneous $k$-facility games; our results are summarized in Table \ref{tab:table-of-results}. 
Firstly, we focus on the case where there is only one facility to be located.
Feigenbaum and Sethuraman~\cite{FJ15} have proven that there is no deterministic 
strategy-proof mechanism with bounded approximation for \util for this case where
the preferences of the agents are known, and their locations are unknown. We study 
the complementary case where the locations are known and the preferences are not
known to the planner. We prove that in this case, the mechanism that places the 
facility on an optimal location for the reported preferences of the agents is 
strategy-proof. In fact, our result is much stronger since it holds for any 
combination of the following relaxations. 
\begin{itemize}
    \item The utility function of every agent can be any function that it is monotone
    with respect the distance between the location of the agent and the location of the
    facility. Thus, if an agent wants to be close to the facility, his utility decreases
    with the distance, and if he wants to be away from the facility, his utility increases.
    \item The domain $D_i$ of every agent's possible locations and the domain $S$ of allowed 
    locations for the facility can be any subset of $\reals^d$. In addition, it can be the case 
    that $S \cap D_i = \emptyset$ for every agent $i$.
\end{itemize}

Next, we focus on the \util objective. We prove that there is no optimal deterministic 
strategy-proof or strategy-proof in expectation mechanism for $k$-facility games 
even for instances with $k=2$, two agents, and known locations for the agents. 
We complement these results by deriving inapproximability bounds for deterministic 
and randomized strategy-proof mechanisms. 
The techniques we use are fundamentally different from~\cite{SV16},
since in our model, the facilities can be located anywhere on the segment without 
any constraint, making the analysis more complex.

Then, we focus on $2$-facility games and we propose strategy-proof mechanisms 
that achieve constant approximation ratio for the \util objective even both locations and preferences are private information.
All of our mechanisms are \emph{simple} and require \emph{limited communication}. By limited communication, we mean that our mechanisms require only a constant number of bit-information from every agent.
To the best of our knowledge, this is the first paper to study the communication complexity 
on facility location problems and how communication affects approximation.
We propose two deterministic and two randomized mechanisms. The first deterministic
mechanism, called \fixed, requires zero communication between the planner and the 
agents. On any instance, \fixed locates the facilities symmetrically away from the middle 
of the segment without requiring any information from the agents. \hl{Although this mechanism 
might seem naive and probably not useful in practice, it achieves constant approximation hence it can be seen as the absolute benchmark for any mechanism.} Furthermore, we prove that \fixed 
is \emph{optimal} when no communication is allowed. No communication means that the agents do not transmit any bits to the planner before the locations for the facilities are decided, or equivalently that the facilities have to be located without getting any information from the agents. The second mechanism, termed 
\fixedp, utilizes the intuition gained from \fixed and chooses between five different 
location-combinations for the facilities and locates the facilities in one of them by using 
the information it got from the agents. Furthermore, every agent has to communicate 
only 5 bits of information to the agent.
Our first randomized mechanism, termed \rand, places with half probability both facilities 
on the beginning of the segment and with half probability both facilities on the end of the segment.
\rand seems naive, but it achieves $\frac{1}{2}$-approximation, it is universally
strategy-proof and requires 
zero communication. \hl{Again, this result can be seen as the benchmark for any randomized mechanism.} The second randomized mechanism, \randp, combines the ideas of 
\rand and \fixedp, it is strategy-proof in expectation and improves upon \rand by requiring again only 5 bits of information per agent.

For the special case where agents' locations are known to the mechanism and all the 
agents are indifferent or want to be close to the facilities, we show how we can 
utilize the optimal mechanism for the 1-facility game and get a 
$\frac{3}{4}$-approximate strategy-proof mechanism for \util when $k=2$.


As a byproduct, we show that \fixed and \rand  achieve the same approximation
guarantee for \happy and \welfare. Thus, we establish lower bounds
that were not known before and complement the results of~\cite{ZL15}. 

\begin{table}[h!]
    \centering
     \begin{tabular}{|c||l|l|c|c|c||c|c|}
        \hline
        \textbf{\#Facilities} & \textbf{Bound} & \textbf{Mechanism} & \textbf{Bits} & \textbf{Preferences} & \textbf{Theorem} & {\bf Loc.} & {\bf Prefs.} \\ \hline \hline
        1 & 1 & OPT-1 & - & $\{-1,0,1\}$ & \ref{thm:one-true} & \cmark & \xmark \\ \hline \hline
        2  & $0.851^{\bf{*}}$ & - & - & $\{-1,0,1\}$ & \ref{alg-rand-inapprox} & \xmark & \cmark \\ \hline \hline
        2 & $0.292$ & \fixed  & 0 & $\{-1,0,1\}$ & \ref{thm:mech2} & \xmark & \xmark \\ \hline
        2 & $0.366$ & \fixedp  & 5 & $\{-1,0,1\}$ & \ref{thm:fixedp-appx}, \ref{thm:fixedp-cc} & \xmark & \xmark \\ \hline
        2 & $0.5$   & \rand & 0  & $\{-1,0,1\}$ & \ref{thm:rand} & \xmark & \xmark \\ \hline
        2 & $0.538$ & \randp & 5 & $\{-1,0,1\}$ & \ref{thm:randp-apx}, \ref{thm:rand-cc} & \xmark & \xmark \\ \hline
        2 & $0.75$ & $OPT^2$ & 0 & $\{0,1\}$ & \ref{thm:opt2} & \cmark & \xmark \\ \hline \hline
        k & $0.5$ & \fixedzo & 0 & $\{0,1\}$ & \ref{thm:fzo} & \xmark & \xmark \\ \hline
        k & $\frac{\lfloor \frac{k}{2}\rfloor}{k}$ & \fixedzm & 0 & $\{-1,0\}$ & \ref{thm:fzm} & \xmark & \xmark \\ \hline
    \end{tabular}
    \caption{Snapshot of our results. The bound $0.851^*$ is an inapproximability result. The column ``bits'' corresponds to the bits per agent each of our mechanism needs. The preferences show the allowed preferences of the agents. An agent has preference $-1$ if he wants to be away from a facility; 1 if he wants to be close to a facility; and 0 if he is indifferent about the facility. The last two columns correspond to the information that is publicly available: ``Loc.'' corresponds to the locations of the agents while ``Prefs.'' corresponds to the preferences of the agents. Signs \cmark and \xmark~indicate whether this information is public or private respectively. \label{tab:table-of-results}}
\end{table}


\subsection{Further related work}
There is a long line of work on homogeneous facility location 
games~\cite{AFPT10,DFMN12,FT10,FT14,LM+19,Lu10,Lu09,M19,ZL14}. 
Different objectives and different utility functions have been 
studied as well. In~\cite{FSY} the objective was the sum of 
$L_p$ norms of agent's utilities, while in~\cite{FW} it was the sum of least 
squares. \cite{FLZZ} introduced double-peaked utility functions. 
The obnoxious facility game on the line, where every agent wants to be 
away from the facilities, was introduced in~\cite{CWZ11} and later the model was 
extended for trees and cycles in~\cite{CWZ13}. In~\cite{YMZ}, the objective of least 
squares for obnoxious agents, was studied. The maximum envy was recently introduced 
as an objective for facility location games in \cite{CFT16}. In
that paper as well as in \cite{GT17}, the authors studied the approximation of
mechanisms according to additive errors. False-name proof mechanisms 
for the location of two identical facilities were studied in \cite{STY16} while
\cite{Th10} gave a characterization of strategy-proof and group strategy-proof 
mechanisms in metric networks for 1-facility games with private locations of the agents.
Since the conference version of this paper~\cite{AD18}, other papers on 
heterogeneous facility location games have appeared. In~\cite{DLLX}, the authors 
studied heterogeneous 2-facility games on a line segment, under the extra constraint
where the locations between the two facilities have to be at least a certain distance.
In~\cite{KVZ}, the authors studied heterogeneous facility location games where the agents
were located on a line but the facility could be placed in a region on the plane.
\hl{Finally,} \cite{li2019strategyproof} \hl{studies a closely related model for 2-facilities 
heterogeneous games under the social welfare objective. In their model, there are two 
facilities, $f_1$ and $f_2$, to be located on the line. Every agent has as private information
his location, and a subset of ``acceptable'' facilities. The} {\em cost} \hl{an agent has, is the minimum distance between his location and the closest acceptable facility. There the objective is to choose locations for the facilities such that the sum of the costs of the agents is minimized.}

Simple mechanisms received a lot of attention lately; see \cite{GN17} for example and 
the references therein for simple auctions. Informally, a simple mechanism is
easy to implement and allows the agents to ``easily'' deduce the strategy-proofness of 
the mechanism. One way to capture simplicity is to use \emph{verifiably truthful} 
mechanisms \cite{BraP15}, where agents can check whether a mechanism is strategy-proof 
by using some, possibly exponential, algorithm. Simple mechanisms were formalized in 
\cite{LiOSP} by introducing \emph{obviously} strategy-proof mechanisms.
\cite{FV17} analysed this type of mechanisms for homogeneous 1-facility games. 

After a long history in theoretical computer science~\cite{KN}, communication complexity problems
have been studied in auction settings~\cite{BNS07} and in facility location games~\cite{feldman2016voting} but with ordinal preferences of the agents as input to the mechanisms.
Communication complexity has also been studied in other more general mechanism design problems
~\cite{MT14,Zandt}. To the best of our knowledge, no one studied the communication complexity
of facility location games on the line with cardinal utilities.

\section{Model}

In a \emph{$k$-facility game}, there is a set $N = \{1, \ldots, n\}$ 
of agents located in $\reals^d$ and a set of $k$ distinct facilities 
$F = \{1, \ldots, k\}$ that need to be placed in $S \subseteq \reals^d$. 
Each agent $i$ is associated with a location $x_i \in \reals^d$ and a vector 
$t_i \in \{-1,0,1\}^k$ that represents his preferences for the facilities. 

If agent $i$ wants to be \emph{far} from facility $j$, then $\tij=-1$; if he
is \emph{indifferent}, then $\tij = 0$; if he wants to be \emph{close} to 
$j$, then $\tij = 1$.
We will use $\ybf = (\ybf_1, \ldots, \ybf_k)$ to denote the locations of
the facilities and $s=(s_1,\ldots, s_n)$ to denote the profile of the agents, 
i.e. their declared tuples $s_i=(x_i,t_i), \forall i \in N$. 
A vector $s_{-i}=(s_1, \ldots, s_{i-1}, s_{i+1}, \ldots, s_n)$ is the vector of 
tuples excluding $s_i$, thus we can denote a profile as $(s_i,s_{-i})$. 

The utility that agent $i$ gets from facility $j$, denoted as $u_{ij}(x_i,t_i,\ybf)$, depends 
on the distance $\dist(x_i, \ybf_j)$ between the location of the agent and the
location of the facility $j$, and on the agent's preference $t_{ij}$ for that facility. 
We assume that $u_{ij}$ follows the rules below:
\begin{itemize}
\item If $t_{ij} = -1$, then $u_{ij}(x_i,t_i,\ybf)$ is strictly increasing with $\dist(x_i,\ybf_j)$.
\item If $t_{ij} = 0$, then $u_{ij}(x_i,t_i,\ybf)$ is a constant independent of $\dist(x_i,\ybf_j)$.
\item If $t_{ij} = 1$, then $u_{ij}(x_i,t_i,\ybf)$ is strictly decreasing with $\dist(x_i,\ybf_j)$.
\end{itemize}
The total utility agent $i$ gets under \ybf is defined as the sum of the utilities 
he gets for each of the facilities, i.e. 
$u_i(x_i, t_i,\ybf)=\sum_{j \in [k]} u_{ij}(x_i,t_i,\ybf_j)$.
We consider three different objective functions:
\util, defined as $\max_\ybf \min_i u_i(x_i, \ti, \ybf)$; \welfare defined
as $\max_\ybf \sum_i u_i(x_i,t_i,\ybf)$; and \happy defined as 
$\max_\ybf \min_i\frac{u_i(x_i,t_i,\ybf)}{u_i^*(x_i,t_i)}$
where $u_i^*(x_i,t_i)=\max_{\ybf}u_i(x_i,t_i,\ybf)$.

A \emph{mechanism} $M$ is an algorithm that takes as input a profile $s$ and 
outputs the locations of the facilities, $\ybf$. 
A mechanism is \emph{deterministic} if it chooses \ybf deterministically and 
it is \emph{randomized} if \ybf is chosen according to a probability distribution.
Let $\opt(s)$ and $M(s)$ denote the optimal value and the value of mechanism $M$
for an objective function under the profile $s$ respectively. A mechanism $M$ achieves
an approximation ratio $\alpha\leq 1$, or it is $\alpha$-approximate, if for any type 
profile $s$, it holds that  $M(s) \geq \alpha \cdot \opt(s)$.
A mechanism is called strategy-proof if no agent can benefit by misreporting \hl{his
location} or his preferences. Formally, a mechanism $M$ is strategy-proof if for any true profile
$(s_i,s_{-i})$ it returns locations \ybf and any misreported profile 
$(s_i',s_{-i})$ it returns $\ybf'$, we have that 
$u_i(x_i, \ti, \ybf) \geq u_i(x_i, \ti, \ybf')$. A randomized mechanism is 
universally strategy-proof if it is a probability distribution over deterministic 
strategy-proof mechanisms and strategy-proof in expectation if no agent can increase 
his \emph{expected} utility by misreporting his type. Furthermore, a mechanism is 
called {\em false-name proof} if no agent can benefit by using multiple and different 
identities in the game. 
The strongest notion of strategy-proofness for a mechanism is to be {\em group
strategy-proof}. \hl{For any subset $Z \subseteq N$ of agents, let $(s_Z, s_{-Z})$ denote 
a profile of the agents' declarations. Furthermore, let $\ybf$ be the output of the
a mechanism under the true types $(s_Z, s_{-Z})$ and let $\ybfp$ be the output of the
the mechanism under $(s'_Z, s_{-Z})$ where agents in $Z$ had coordinated their 
declarations.
A mechanism $M$ is group strategy-proof, if for any $Z \subseteq N$, any $i \in Z$, and any $s'_Z \neq s_Z$ it holds that $u_i(x_i, \ti, \ybf) \geq u_i(x_i, \ti, \ybfp)$.}

\paragraph{\bf Communication Complexity.}\hl{ The communication complexity of a mechanism is the number of bits each agent has to send to the mechanism in order to compute the output. We say that a mechanism has zero communication complexity if it requires 0 bits from every agent.}

\subsection{Facility location on a line segment}
A special case of $k$-facility games is when all the agents are located on the line 
segment $[0, \ell]$, where $\ell > 0$. 
This case is studied in the literature~\cite{PT09,SV16} since the definitions
above are greatly simplified. For normalization purposes, we assume that the 
maximum utility an agent can get from any facility is $\ell$, and we define
the utility function of agent $i$ as follows.
\noindent
\begin{align}
\label{eq:util}
u_{ij}(x_i,t_i ,\yj)=\left\{ \begin{array}{rl}
|x_i - \ybf_j|, & \quad \text{if $\tij=-1$}\\
\ell, & \quad \text{if $\tij=0$}\\
\ell-|x_i - \ybf_j|, & \quad \text{if $\tij=1$.} \end{array}\right.
\end{align}



\section{1-facility games with known locations}
\label{sec:onef}
We first study the case where the locations of the agents are publicly 
known and only one facility has to be placed. 
We will show that the mechanism which places the facility on an optimal location 
using any {\em declaration-independent} tie-breaking rule is strategy-proof
for \util, \welfare, \happy objectives.

\begin{definition}
A mechanism $M$ has a declaration-independent tie-breaking rule if it outputs 
the same $\ybf$ for any two profiles $s \neq s'$  with $M(s) = M(s')$.
\end{definition}

Hence, a mechanism has a declaration-independent tie-breaking rule if it outputs 
the same location for the facility for all profiles that yield the same value 
for the objective we are trying to optimise. An example of such a rule is the 
lexicographic minimum.


\alg{alg:onef}
\begin{tcolorbox}[title=OPT-1 Mechanism]
\begin{itemize}
\item[{\bf In:}] For every agent $i$: public location $x_i \in \reals^d$, 
private preference $t_i  \in \{-1,0,1\}$; region $S \subseteq \reals^d$; 
objective \ocal; declaration-independent tie-breaking rule $T$.
\item[{\bf Out:}] Location $\ybfs \in S$ for the facility.
\end{itemize}
\begin{enumerate}
\item Let $Y \subseteq S$ such that every $y \in Y$ optimizes \ocal for the given locations 
and preferences, excluding the agents with preference 0.
\item Choose $\ybfs \in Y$ according to the tie-breaking rule $T$.
\end{enumerate}
\end{tcolorbox}

Mechanism~\ref{alg:onef} does not make any assumptions about the dimensions of the 
agents locations and the region $S$. So, the actual locations of the agents can be 
in $\reals^{d_1}$ and the region $S \subseteq \reals^{d_2}$, where $d_1 \neq d_2$. 
In addition, $S$ can be of an arbitrary form, i.e. it can be the union of several
disjoint regions of $\reals^d$. 

\subsection{Analysis for \util objective}
\label{sec:onef-util}
In this section we focus on the \util objective, i.e. $\ocal = \max_y \min_i u_i(x_i,t_i,y)$. 
In order to prove that Mechanism~\ref{alg:onef} is strategy-proof for \util, 
we partition the agents into two sets \caltl and \calth.
\caltl contains the agents with the minimum utility when the facility is placed
on \ybfs and $\calth = N \setminus \caltl$.
Since agents with preference type 0 have constant utility independently of $\ybf$ and 
they are excluded from the computation of $\ybf$, in our analysis we will assume that
there is no agent $i$ with $t_i=0$. We first prove that no agent from the set \calth has an incentive to lie.
\begin{lemma}
\label{lem:one-high}
No agent from \calth can increase his utility by lying.
\end{lemma}
\begin{proof}
For the sake of contradiction suppose that an agent $i \in \calth$ with 
preference $\ti$ declares preference $\ti'$ and increases his utility.
Let \ybfp be the optimal location of the facility in this case.
Since we have assumed that agent $i$ increases his payoff, we have that
$u_i(x_i, \ti, \ybfs) < u_i(x_i, \ti, \ybfp)$.
We will consider two cases depending on the declaration $\ti'$.
\begin{itemize}
\item $\ti' = 0$. Recall, in this case, Mechanism~\ref{alg:onef} excludes 
agent $i$ from the computation of $\ybfp$. 
Since $u_i(x_i, \ti, \ybfs) < u_i(x_i, \ti, \ybfp)$, we get that $\ybfs \neq \ybfp$.
In addition, we get that 
$\min_{j \neq i} u_j(x_j, t_j, \ybfp) > \min_{j \neq i} u_j(x_j, t_j, \ybfs)$; if
this was not the case, the mechanism could return \ybfs and increase the value of
the objective. Hence, we get that 
$\min_{j} u_j(x_j, t_j, \ybfp) > \min_{j} u_j(x_j, t_j, \ybfs)$.
This means that $\ybfp$ is a better solution than $\ybfs$ for the \util objective, which
contradicts the assumption that \ybfs is an optimal solution.
\item $t_i' \neq 0$. The utility of agent $i$ will change only if the location of 
the facility changes; this is due to the declaration-independent tie-breaking rule $T$. 
This will happen only if $u_i(x_i, \ti', \ybfs) < \min_{j \neq i} u_j(x_j, t_j, \ybfs)$.
This means that $u_i(x_i, \ti', \ybfs) < u_i(x_i, \ti', \ybfp)$.
Without loss of generality let $\ti' = 1$.
In this case $\dist(x_i,\ybfs) > \dist(x_i,\ybfp)$, i.e. the new optimal location is 
closer to $x_i$. But this means that $u_i(x_i, \ti, \ybfs) < u_i(x_i, \ti, \ybfp)$ 
and since $\ti = -1$ this contradicts the assumption that agent $i$ can increase his
utility by misreporting his preference.
\end{itemize}
\end{proof}
Next, we prove that no agent from \caltl has an incentive to lie about his
preferences.
%
\begin{lemma}
\label{lem:one-low}
No agent from \caltl can increase his utility by lying.
\end{lemma}
\begin{proof}
We will prove the claim by contradiction. Suppose 
that an agent $i \in \caltl$ with preference $\ti$ can increase his utility by 
declaring $\ti'$.
Using exactly the same arguments as in Lemma~\ref{lem:one-high} we can see that
$\ti' \neq 0$.
Let $\ybfp \neq \ybfs$ be the optimal location for the facility when agent $i$
declares $\ti'$. Clearly, if $\ybfp = \ybfs$ agent $i$ has no reason to lie.
We now consider the following two cases:
\begin{itemize}
\item $u_i(x_i, \ti', \ybfs) \geq u_i(x_i, \ti, \ybfs)$.
Since we have assumed that agent $i$ increases his utility by declaring $\ti'$, 
we have that $u_i(x_i, \ti, \ybfp) > u_i(x_i, \ti, \ybfs)$. Hence,
$\min_{j \neq i} u_j(x_j, t_j, \ybfp) > \min_{j \neq i} u_j(x_j, t_j, \ybfs)$
since an agent $j \ne i$ who now has the minimum utility is the one who determines the 
new outcome $\ybfp$. We note that $\min_{j \in N}$ $u_j(x_j, t_j, \ybfp)$ should 
be strictly larger than $\min_{j \in N} u_j(x_j, t_j, \ybfs)$ since  
Mechanism~\ref{alg:onef} uses a declaration-independent tie-breaking rule.
So the location should not change if the value of the objective remains the same.
But then we have that $\min_{j \in N} u_j(x_j, t_j, \ybfp) >
\min_{j \in N} u_j(x_j, t_j, \ybfs)$
which contradicts the fact that \ybfs is an optimal location for the facility.

\item $u_i(x_i, \ti', \ybfs) < u_i(x_i, \ti, \ybfs)$.
This means that agent $i$ under the declaration $\ti'$ has the smallest
utility over all the agents.
Hence, one of the following cases must be true since we assumed that $\ti' \neq 0$.
The first one is when $\ti = -1$ and $\ti' = 1$. Since the utility of agent $i$
under the declaration $\ti'$ increased, it means that $\dist(x_i,\ybfp) < \dist(x_i,\ybfs)$,
i.e. the facility must be placed \emph{closer} to his location $x_i$. But this means that 
his utility under the true preference $\ti$ decreased because the agent wants
to be away from the facility.
Similarly when $\ti = 1$ and $\ti' = -1$ the facility must be placed further away from the position of the
agent, while the agent wants to be close to the facility. Hence, in both cases the utility
of agent $i$ decreases.
\end{itemize} 
As a result, in every case agent $i$ cannot increase his
utility by lying, which contradicts our assumption.
\end{proof}
Notice that Mechanism~\ref{alg:onef} places the facility on the location that
maximizes our objective, i.e. it is optimal. Furthermore, the combination of
Lemmas~\ref{lem:one-high} and~\ref{lem:one-low} shows that no agent can increase
his utility by lying. The next theorem follows:

\begin{theorem}
\label{thm:one-true}
OPT-1 is an optimal strategy-proof mechanism for the \util objective.
\end{theorem}

Theorem~\ref{thm:one-true} complements in a sense the result of~\cite{FJ15}, where it was
proven that there is no deterministic strategy-proof mechanism with 
bounded approximation for the \util objective for 1-facility games even on a line segment with 
known preferences but unknown locations. 


\subsection{Analysis for the \welfare objective}
\label{sec:onef-welfare}
In this section we focus on the \welfare objective, i.e. $\ocal = \max_y \sum_i u_i(x_i,t_i,y)$.
Again, since agents with preference type 0 have constant utility independently of $y$,
we will assume that there is no agent $i$ with $t_i=0$.

\begin{theorem}
\label{thm:one-true-welfare}
Mechanism~\ref{alg:onef} is an optimal strategy-proof for the \welfare objective.
\end{theorem}
\begin{proof}
We will prove the theorem by contradiction. So, assume that
there exists an agent $i$ who can increase his utility by declaring $t_i' \neq t_i$.
Let $\ybfs$ be the optimal location of the facility when $i$ declares $t_i$ and 
$\ybfp \neq \ybfs$ be the location of the facility when he declares $t'_i$. 
So, by assumption, we have that $u_i(x_i,t_i,\ybfs) < u_i(x_i,t_i,\ybfp)$.

Firstly, assume that $t_i' = 0$. Then, Mechanism~\ref{alg:onef} excludes agent $i$ from 
the computation of $\ybfp$. In addition, since the mechanism uses a 
declaration-independent tie-breaking rule and $\ybfp \neq \ybfs$, it must be true that 
$$\sum_{j \neq i}u_j(x_j,t_j,\ybfs) < \sum_{j \neq i}u_j(x_j,t_j,\ybfp).$$
If this was not the case, we could increase the value of the objective by choosing
$\ybfs$ instead.
Thus, since we assumed that $u_i(x_i,t_i,\ybfs) < u_i(x_i,t_i,\ybfp)$, we get that 
$\sum_{j}u_j(x_j,t_j,\ybfs) < \sum_{j}u_j(x_j,t_j,\ybfp)$ which contradicts the 
assumption that \ybfs maximizes the social welfare. 

Having established that $t_i' \neq 0$, we consider the following two cases depending on the 
utilities of the rest of the agents under \ybfs and \ybfp. In what follows we will assume
that $t_i = 1$ and $t_i'=-1$; the arguments for $t_i = -1$ and $t_i'=1$ are similar.
\begin{itemize}
\item $\sum_{j \neq i}u_j(x_j,t_j,\ybfs) < \sum_{j \neq i}u_j(x_j,t_j,\ybfp)$. Then,
as above, we get that \ybfs does not maximize the welfare objective since we have 
assumed that $u_i(x_i,t_i,\ybfs) < u_i(x_i,t_i,\ybfp)$.
\item $\sum_{j \neq i}u_j(x_j,t_j,\ybfs) \geq \sum_{j \neq i}u_j(x_j,t_j,\ybfp)$. 
Since $\ybfs \neq \ybfp$ and since Mechanism~\ref{alg:onef} has a declaration-independent
tie-breaking rule, it should be true that 
$$\sum_{j \neq i}u_j(x_j,t_j,\ybfs) +  u_i(x_i,t'_i,\ybfs) < \sum_{j \neq i}u_j(x_j,t_j,\ybfp)
+ u_i(x_i,t'_i,\ybfp).$$
So, we get that $u_i(x_i,t'_i,\ybfs) < u_i(x_i,t'_i,\ybfp)$ and since we have assumed that
$t_i' = -1$ we get that $\dist(x_i, \ybfs) < \dist(x_i, \ybfp)$. This, in turn means that
$u_i(x_i,t_i,\ybfs) > u_i(x_i,t_i,\ybfp)$ which is a contradiction.
\end{itemize}
Hence, we have shown that for any declaration $t_i'$ the utility of the agent cannot increase.
Thus, the theorem follows.
\end{proof}

\subsection{Analysis for the other objectives}
Observe that in the analyses in Sections~\ref{sec:onef-util} and~\ref{sec:onef-welfare},
the only assumption about the utility functions of the agents is that they are monotone
with respect to the distance between the location of the agent and the location of the 
facility. Thus, every agent can have his own type of utility function, completely different
than the types of the other agents.

Recall, for \happy  we have that 
$\ocal = \max_\ybf \min_i\frac{u_i(x_i,t_i,\ybf)}{u_i^*(x_i,t_i)}$ where
$u_i^*(x_i,t_i)=\max_{\ybf}u_i(x_i,t_i,\ybf)$. Observe though, $u_i^*(x_i,t_i)$ is
a constant, hence the analysis of Section~\ref{sec:onef-util} applies here as well. 
So, we can get the following as a corollary of Theorem~\ref{thm:one-true}.
\begin{corollary}
\label{cor:one-true-happy}
Mechanism~\ref{alg:onef} is an optimal strategy-proof for the \happy objective.
\end{corollary}

\section{Inapproximability results}
\label{sec:inapprox}

In the remainder of the paper, unless specified otherwise, we study the \util objective.
In this section, we provide inapproximability results for strategy-proof 
mechanisms for 2-facility games. We show that the second facility changes dramatically 
the landscape of strategy-proofness. We prove that the extension of the optimal mechanism 
for two facilities, i.e. placing the facilities on the locations that maximize the objective 
under the declared preferences of the agents, is not strategy-proof even in the setting
of a line segment with two agents and known locations. 
Furthermore, we provide inapproximability results for strategy-proof mechanisms. 


We first prove that there is no 0.851-approximate deterministic strategy-proof 
and then extend it to strategy-proof in expectation mechanisms.

\begin{theorem}
\label{thm:2fub}
There is no $\alpha$-approximate deterministic strategy-proof mechanism for the
2-facility game with $\alpha \geq 0.851$.
\end{theorem}
\begin{proof}
Let us consider the instances $I$ and $I'$ depicted in Figure~\ref{fig:fig1}. 
Each white circle corresponds to an agent. Agent $a_1$ is located on 0 and agent $a_2$ on 
$x > \frac{2\ell}{3}$, where the exact value of $x$ will be specified later in the proof.
Without loss of generality, we assume that $\ell = 1$. Firstly, we will prove that the
mechanism that places the facilities on their optimal locations is not strategy-proof
even when the locations of the agents are known. Then, we will use these instances to 
derive our inapproximability result.

On instance $I$ agents $a_1$ and $a_2$ have preferences $t_1=(-1,1)$ and
$t_2=(0, 1)$ respectively.
It is not hard to see that the optimal locations for the facilities are 
$\ybf_1 = 1$ and $\ybf_2=\frac{x}{2}$ where each agent gets utility 
$2-\frac{x}{2}$. The optimal locations of the facilities are depicted by black circles 
in the figure.

On instance $I'$ agent $a_1$ has the same preferences as on instance $I$ 
while the preferences of agent $a_2$ are $t'_2=(-1,1)$. The optimal locations 
for the facilities in this instance are $\ybf_1 = 1$ and $\ybf_2 = x$ where 
each agent gets utility $2 - x$. 
\noindent
\begin{figure}[h!]
\begin{center}
\subfigure[Instance I]{
\begin{tikzpicture}[thick, scale=0.5]
  \tikzstyle{every node}==[fill=white,minimum size=4pt,inner sep=0pt]

\draw (-4.4,-0.7) node(v)[label=below:$-1 1$]{};
\draw (-4,0.7) node(v1)[label=above:$0$]{};
\draw (-0.5,0.7) node(v2)[label=above:$\frac{x}{2}$]{};
\draw (3,0) node(v3)[draw, fill=white, circle]{};
\draw (-0.5,0) node(v3)[draw, fill=black, circle]{};
\draw (6,0) node(v3)[draw, fill=black, circle]{};
\draw (-4,0) node(v7)[draw, fill=white, circle]{};
\draw (-0.5,-1.9) node(u)[label=below:$\ybf_2$]{};
\draw (3,-0.7) node(u1)[label=below:$0 1$]{};
\draw (5.2,-1.9) node(l1)[label=below right:$\ybf_1$]{};
\draw (6,0.7) node(u1)[label=above :$\ell$]{};
\draw (3,0.7) node(v6)[label=above:$x$]{};
\draw (-3.88,0) -- (2.9,0);
\draw (3.15,0) -- (6,0);
\node[] at (7,-2) {};
\end{tikzpicture}
}
\subfigure[Instance $I'$]{
\begin{tikzpicture}[thick, scale=0.5]
  \tikzstyle{every node}==[fill=white,minimum size=4pt,inner sep=0pt]

\draw (-4.4,-0.7) node(v)[label=below:$-1 1$]{};
\draw (-4,0.7) node(v1)[label=above:$0$]{};
\draw (3,0) node(v3)[draw, fill=black, circle]{};
\draw (-4,0) node(v7)[draw, fill=white, circle]{};
\draw (6,0) node(v7)[draw, fill=black, circle]{};
\draw (3,-1.9) node(u)[label=below:$\ybf'_2$]{};
\draw (2.6,-0.7) node(u1)[label=below:$-1 1$]{};
\draw (5.6,-1.9) node(l1)[label=below right:$\ybf'_1$]{};
\draw (5.85,0.7) node(u1)[label=above :$\ell$]{};
\draw (3,0.7) node(v6)[label=above:$x$]{};
\draw (-3.83,0) -- (5.85,0);
\node[] at (7,-2) {};
\end{tikzpicture}
}
\caption{Example for preferences in $\{-1,0,1\}^2$.}
\label{fig:fig1}
\end{center}
\end{figure}
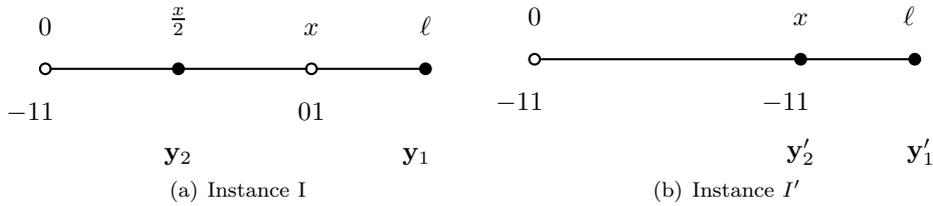
\noindent

Instances $I$ and $I'$ show that the mechanism which places the
facilities on the optimal locations is not strategy-proof. On instance $I$
agent $a_2$ can declare $t'_2=(-1,1)$ and increase his utility from 
$2-\frac{x}{2}$ to $2$. 

Next, we focus on the inapproximability result for any deterministic mechanism.
The high-level idea of the proof is as follows. We assume that we know a strategy-proof
mechanism $M$ that achieves the best possible approximation ratio for the problem. Firstly, 
we focus on instance $I$ where we show that $M$ always places the first facility on 1 and we 
derive the approximation guarantee of $M$ on $I$ as a function of the location $\ybf_2$ 
of the second facility. Then, we turn our attention to instance $I'$, and we observe that
for the location $\ybf'_2$ of the second facility in this case, it should be true that
$\ybf'_2 \leq \ybf_2$. Using this, we consider the two possible cases for the location
$\ybf_1$ of facility $f_1$ with respect to $x$ and we derive bounds on the approximation
ratio of $M$, in each case as a function of $x$. Then, we optimize the value of $x$ and
derive the claimed bound.

So, let $M$ be a strategy-proof mechanism that achieves the best possible approximation
for the \util objective.
We first argue that on instance $I$ mechanism $M$ should place facility $f_1$ on $z=1$.
If this was not the case, the utility of $a_1$ would strictly increase by the movement 
of $f_1$ to $1$, while the utility of $a_2$ would remain the same. 
Hence, the approximation ratio of $M$ would strictly improve by placing $f_1$ on 1, 
contradicting the assumption that $M$ achieves the best approximation guarantee.

Next, suppose that $M$ places facility $f_2$ on $\ybf_2 \leq x$ on instance $I$.
Since $M$ is strategy-proof, $f_2$ cannot be placed on any 
$\ybf'_2 > \ybf_2$ on instance $I'$. If $\ybf'_2 > \ybf_2$, then agent $a_2$ 
from $I$ could declare preferences $t'_2 = (-1,1)$ and increase his utility
(assuming that $ x > \frac{2\ell}{3}$). 
We consider the following two cases regarding the location $\ybf'_1$ in which 
$M$ places $f_1$ on $I'$:

\begin{itemize}
\item $\ybf'_1 \geq x$. Then, obviously $\ybf_1 = 1$ since otherwise the 
utility of both agents in $I'$ is decreasing and thus $M$ does not achieve the 
maximum approximation. So, under $M$ agent $a_2$  on instance $I'$ gets 
utility at most $u_2' = 2-2x + \ybf_2$ while $a_2$ gets utility $u_1' = 2 - \ybf_2 \geq u_2'$
(since $x \geq \ybf_2$).
Thus $M$ achieves an approximation of $\frac{2-2x + \ybf_2}{2-x}$. 
Furthermore, on instance $I$ agent $a_1$ gets utility $u_1 = 2 - \ybf_2$,
since as explained earlier, $M$ places $f_1$ on 1, while $a_2$
gets utility $u_2 = 2 -x + \ybf_2 \geq u_1$ when $\ybf_2 \geq \frac{x}{2}$. 
Clearly if $\ybf_2 < \frac{x}{2}$ the utility of $a_1$ gets worse. Thus, the 
approximation of $M$ on instance $I$ is $\frac{4 - 2\ybf_2}{4-x}$.
Observe that the approximation guarantee of $M$ on $I$ is decreasing with 
$\ybf_2$ while on $I'$ it is increasing with $\ybf_2$.
So, if we optimize the approximation guarantee and solve for $\ybf_2$ we get 
that $\ybf_2 =\frac{6x-2x^2}{8-3x}$. Thus, if $\ybf'_1 > x$, the approximation 
of $M$ is at most
\begin{align}
\label{eq:case1}
\frac{4-2\cdot \frac{6x-2x^2}{8-3x}}{4 -x} = \frac{4x^2-24x+32}{3x^2-20x+32}
\end{align}
\item
If $M$ on instance $I'$ places $f_1$ on $\ybf'_1 < x$, then observe that there 
is no location $\ybf'_2$ for $f_2$ such that both agents get utility strictly 
larger than 1.
Thus, in this case $M$ achieves approximation at most 
\begin{align}
\label{eq:case2}
\frac{1}{2 -x}
\end{align}
\end{itemize}
Observe that the approximation guarantee in~\eqref{eq:case1} increases with $x$ 
while in~\eqref{eq:case2} it decreases with $x$. So if we optimize on the approximation guarantee of $M$,
we have to solve for $x$ the equation $-4x^3+29x^2-60x+32=0$. The unique 
solution in $[0,1]$ is $x=\frac{13-\sqrt{41}}{8}$.
Using this value in~\eqref{eq:case1} and~\eqref{eq:case2}
we get that any deterministic strategy-proof mechanism on instances $I$ and $I'$ achieves approximation less than 0.851.
\end{proof}

The inapproximability bound 
can be extended to strategy-proof in expectation mechanisms.
\begin{theorem}\label{alg-rand-inapprox}
There is no $\alpha$-approximate strategy-proof in expectation mechanism for the
2-facility game with $\alpha \geq 0.851$.
\end{theorem}

\begin{proof}
We will use again the instances from Figure~\ref{fig:fig1} to prove the claim setting $x = \frac{13-\sqrt{41}}{8}$. Recall that the optimal utility on instance $I$ is $\frac{4-x}{2}$ and on $I'$ it is $2-x$.

So, let $M$ be a  strategy-proof in expectation mechanism. 
Observe that on instance $I$ the mechanism should place the facility $f_1$ on 1
for the same reason as the one mentioned in the proof of Theorem~\ref{thm:2fub}; every 
other location for $f_1$ decreases the approximation guarantee of $M$. Suppose now that $M$ places $f_2$ on $y \in [0,1]$ 
according to the probability distribution $p(y)$. 
Without loss of generality we can assume that $p(y) = 0$ for every 
$y > x$. This is because the approximation guarantee of $M$ can increase if 
we place the facility on $x$ instead of some $y>x$. 
Hence, on instance $I$ under $M$ agent $a_1$ gets utility 1 from $f_1$ and
utility $\int_0^x p(y)(1-y)dy = 1 - \int_0^x p(y)y dy$ from facility $f_2$, so $u_1 = 2-\int_0^x p(y)ydy$ in total.
Similarly, agent $a_2$ gets utility 1 from $f_1$ and utility 
$1-x + \int_0^x p(y)ydy$ from facility $f_2$, so $u_2 = 2 - x + \int_0^x p(y)ydy$ in total.
Then, since $u_2 < u_1$, the approximation guarantee of $M$ on $I$ is at most 
\begin{align}
\label{eq:ra1}
\frac{2}{4-x}\cdot\left(2-\int_0^x p(y)ydy\right)
\end{align}


We now consider two cases according to the location in which $M$ places facility $f_1$ on instance $I'$.
If $M$ places $f_1$ on $y'_1 \geq x$, then without loss of generality we can
assume that $f_1$ is placed on 1 since every other location decreases the utility
of both agents. 
So suppose that $M$ places $f_1$ on 1 with some probability.

Furthermore, suppose that $M$ places $f_2$ on $y$ according to the probability 
distribution $\pi(y)$ when $f_1$ is placed on 1. Observe that we can assume 
that $M$ does not place $f_2$ on $y > x$, since the utility of both
agents could increase by placing it on $x$ instead. Thus, on instance $I'$, agent $a_2$ gets utility $1-x$ from facility $f_1$
and utility $1-x + \int_0^x \pi(y)ydy$ from facility $f_2$, so $u_2' = 2 - 2x + \int_0^x \pi(y)ydy$ 
in total. Similarly, agent $a_1$ gets total utility $u_1' = 2 -\int_0^x \pi(y)ydy > u_2'$. Since $M$ is
strategy-proof it must hold that $\int_0^x \pi(y)ydy \leq \int_0^x p(y)ydy$.
If this was not the case, then agent $a_2$ from instance $I$ could declare
preferences $(-1,1)$ and increase its utility. As a result the approximation guarantee of
$M$ on $I'$ is at most 
\begin{align}
\label{eq:ra2}
\frac{1}{2-x}\cdot\left(2-2x+\int_0^x p(y)ydy\right)
\end{align}
$M$ achieves the best approximation on both instances when 
the quantities from~\eqref{eq:ra1} and~\eqref{eq:ra2} are equal. Hence, if we 
equalize them and solve for the integral we get that $\int_0^x p(y)ydy = 
\frac{6x-2x^2}{8-3x}$ and the approximation guarantee is less than 0.851 on both
instances for the chosen $x$.

If the mechanism places $f_1$ on $y'_1 < x$, then on any location for $f_2$ 
there will be an agent with utility at most 1 and the approximation guarantee
of the mechanism will be at most $\frac{1}{2-x} < 0.851$. 
Thus, in all possible cases the approximation of $M$ is upper bounded by 0.851.
\end{proof}


\section{Deterministic Mechanisms}
In this section, we propose deterministic strategy-proof mechanisms. 
An initial approach would be to consider each facility independently and place 
it to its optimal location. As we have already proved this mechanism is strategy-proof
when the locations of the agents are public information.
%
However, it achieves poor approximation if the agents 
want to be away from the facilities. Consider the case with
$n$ agents located on $0, \frac{2\ell}{n}, \frac{3\ell}{n},
\ldots, \frac{(n-1)\ell}{n}, \ell$ and each having preferences $(-1,-1)$. 
Observe that the optimal location for one facility is to be placed on 
$\frac{\ell}{n}$ since this location maximizes the minimum distance between any 
agent and the facility. Thus, both facilities will be placed on the same 
location $\frac{\ell}{n}$. 
Then the agent located in 0 has utility $\frac{2\ell}{n}$, the minimum over all 
the agents. It is not hard to see that an optimal solution is to place
facility $f_1$ on 0 and facility $f_2$ on $\ell$ resulting in a utility of $\ell$
for each agent. Hence, the mechanism that places the facilities independently 
to their optimal locations is $\frac{2}{n}$-approximate.

The example above provides evidence that a mechanism with good approximation 
ratio should not put both facilities on the same location if
there are agents who have preference -1 for both facilities; in the worst case 
there is an instance with an agent located in the exact same location where the
facilities are placed and with preferences $(-1,-1)$ resulting in a zero approximation. 
On the other hand, the facilities should not be placed far away from each other. This 
is because, in the worst-case again, an agent might have preference -1 for the 
facility that is close to his location and preference 1 for the facility that is 
far from him.

Using the intuition gained from the discussion above we propose a mechanism for 
the 2-facility game that comprises these ideas and places the 
facilities symmetrically away from the endpoints of the segment and it is strategy-proof even if the locations of the agents is private information. 
Mechanism \fixed depicts our approach. It does not use any information from the
agents, thus it is de facto strategy-proof. 

\begin{definition}[\fixed Mechanism]
Let $z_f = 1-\frac{\sqrt{2}}{2}$. \fixed mechanism sets $\ybf_1=z_f \cdot \ell$ 
and $\ybf_2=(1-z_f)\cdot \ell$.
\end{definition}

\begin{theorem}
\label{thm:mech2}
\fixed  is $z_f \simeq 0.292$-approximate.
\end{theorem}

\begin{proof}
Tables~\ref{tab:main-low} and~\ref{tab:main-high} show the utility the agent located on $x_i$ gets under 
$\ybf=(z\cdot\ell, (1-z)\cdot\ell)$ and the corresponding ratio. 
Our goal is to find a $z \in [0,\ell]$ that maximizes the minimum ratio.
Thus, the optimal guarantee for \fixed is achieved
when $\frac{z}{\ell} = \frac{\ell - 2z}{2\ell - 2z}$. If we solve for
$z$, the feasible solution is $z_f = (1 - \frac{\sqrt{2}}{2})\ell$ and the approximation
guarantee follows.

Finally, observe that if the number of facilities to be placed is at least two, then 
$\max_\ybf \min_i u_i(x_i, \ti, \ybf) \geq 
\max_\ybf \min_i \frac{u_i(x_i, \ti, \ybf)}{u^*_i(x_i, \ti)}$, since 
$u^*_i(x_i, \ti) \geq \ell$. Thus, \fixed can be used for both \util and \happy 
objectives and since it does not use any information from 
the agents, it possesses all the desirable properties like group strategy 
proofness and false name proofness.

\begin{table}[h!]
    \begin{minipage}{0.45\textwidth}
        \centering
        \begin{tabular}{|c|c|c|l|}
            \hline
            $t_i$ & $u_i(x_i,\ti,\ybf)$ & $u^*_i(x_i,\ti)$ & Ratio \\ \hline
            1, 1  & $\ell + 2x_i$  & $2\ell$ & $\geq 1/2$ \\ \hline
            -1, 1 & $2z\cdot\ell$ & $2\ell - x_i$ & $\geq z$ \\ \hline
            1, -1 & $(2 - 2z)\cdot\ell$ & $2\ell - x_i$  & $\geq 1/2$ \\ \hline
            -1, -1 & $\ell - 2x_i$ & $2\ell -2x_i$ & $\geq \frac{1 - 2z}{2 - 2z}$\\
            \hline
        \end{tabular}
        \caption{Case analysis when $x_i \leq z\cdot\ell$ or $x_i \geq (1-z)\cdot\ell$.}
        \label{tab:main-low}
    \end{minipage}
    \hfill
    \begin{minipage}{0.45\textwidth}
        \centering
        \begin{tabular}{|c|c|c|l|}
            \hline
            $t_i$ & $u_i(x_i,\ti,\ybf)$ & $u^*_i(x_i,\ti)$ & Ratio \\ \hline
            1, 1 & $(1 + 2z)\cdot\ell$  & $2\ell$ & $\geq 1/2$ \\ \hline
            -1, 1 & $2x_i$ & $2\ell - x_i$ & $\geq \frac{2z}{2-z}$ \\ \hline
            1, -1 & $2\ell - 2x_i$ & $2\ell - x_i$ & $\geq 2/3$ \\ \hline
            -1, -1 & $(1 - 2z)\cdot\ell$ & $2\ell - 2x_i$ & $\geq \frac{1 - 2z}{2 - 2z}$ \\ \hline
        \end{tabular}
        \caption{Case analysis when $z\cdot\ell<x_i<(1-z)\cdot\ell$.}
        \label{tab:main-high}
    \end{minipage}
\end{table}

\end{proof}

Theorem~\ref{thm:mech2} shows the sharp contrast between 1-facility and 2-facility
games where both locations and preferences are private. Recall that 
~\cite{FJ15} proved that for 1-facility games there is no deterministic
strategy-proof mechanism with bounded approximation guarantee.
Observe furthermore that \fixed does not require any information from the agents. 
Next, we prove that it is optimal when no communication is allowed.

\begin{theorem}
\label{thm:cc-lb}
\fixed is the optimal deterministic mechanism when no communication is allowed.
\end{theorem}

\begin{proof}
Let $M$ be any deterministic mechanism that places the facilities with no communication. 
Since $M$ is deterministic, it places them on the same locations for any instance.
So, let $\ybf_1\cdot \ell$ and $\ybf_2\cdot \ell$ be the locations of the first and the second 
facility respectively. Without loss of generality assume that $0\leq \ybf_1\leq \ybf_2\leq 1$. 
We will prove our claim by contradiction. So, for the sake of contradiction assume that the
approximation ratio of $M$ is strictly better than $z=(1-\frac{\sqrt{2}}{2})$.
Without loss of generality we assume that  $\ybf_1 \leq \frac{1}{2}$. Consider the following two instances.
On the first instance there is only one agent on $\ybf_1\cdot\ell$ with preferences $(-1,-1)$. 
The utility of the agent under $M$ is $(\ybf_2-\ybf_1)\cdot \ell$. The optimal solution places both 
facilities on $\ell$ and the agent gets utility $(2-2\ybf_1)\cdot \ell$. So, the approximation ratio 
of $M$ is $\frac{\ybf_2-\ybf_1}{2-2\ybf_1}$. Since the approximation of $M$ is strictly greater than 
$z$, we get that
\begin{equation}
\label{eq:cc-lb1}
\ybf_1 < \frac{\ybf_2-2z}{1-2z}
\end{equation}
Now, consider the instance where there is only one agent on 0 with preferences $(-1,1)$.
Under $M$, the agent gets utility $(1+\ybf_1-\ybf_2)\cdot\ell$. The optimal solution for this instance 
places the first facility on $\ell$, the second one on 0, and the agent gets utility $2\ell$.
Hence, the approximation guarantee of $M$ on this instance is $\frac{1+\ybf_1-\ybf_2}{2}$.
Again, since we assume that the approximation is strictly greater than $z$, we get that
\begin{equation}
\label{eq:cc-lb2}
\ybf_1 > 2z+\ybf_2-1
\end{equation}
The combination of Equations~\eqref{eq:cc-lb1} and~\eqref{eq:cc-lb2} dictates that
$\ybf_2 > 3-\frac{1}{2z}-2z > 1-z$.
Similarly using another two instances, we can prove that $\ybf_1< z$. More 
specifically, we use the instance where there is only one agent on $\ybf_2\cdot \ell$ with preferences $(-1,-1)$ and the instance where there is only one agent on $\ell$ with preferences $(1,-1)$. Finally, consider again the instance where there is only one agent on 0 with preferences $(-1,1)$. Recall that the approximation guarantee of the mechanism on this instance is $\frac{1+\ybf_1-\ybf_2}{2}$.
So, since $\ybf_1<z$ and $\ybf_2>1-z$, we get that the approximation guarantee is strictly smaller than $z$
which is a contradiction. Our claim follows.
\end{proof}

\subsection{\fixedp mechanism}
In order to  describe \fixedp, we need to introduce the following events:
\begin{itemize}
\item $L_j$: Every agent wants facility $j$ below $\ell/2$. Formally, for every 
agent $i$ with $x_i \leq \frac{\ell}{2}$ it holds that $t_{ij} \in \{0,1\}$ and
for every agent $i$ with $x_i > \frac{\ell}{2}$ it holds that $t_{ij} \in \{0,-1\}$.
\item $H_j$: Every agent wants facility $j$ above $\ell/2$. Formally, for every 
agent $i$ with $x_i \leq \frac{\ell}{2}$ it holds that $t_{ij} \in \{0,-1\}$ and
for every agent $i$ with $x_i > \frac{\ell}{2}$ it holds that $t_{ij} \in \{0,1\}$.
\end{itemize}

\begin{tcolorbox}[title=\fixedp mechanism]
\textbf{Input:} Locations $x_1, \ldots, x_n$ and preferences $p_1, \ldots, p_n$.\\
\textbf{Output:} Locations $\ybf_1$ and $\ybf_2$.\\
Set $z_d= \frac{7}{22} \approx 0.31$.
\begin{enumerate}
\item \label{step:one} 
If events $L_1$ and $L_2$ occur, then set $\ybf_1=\ybf_2=z_d\cdot \ell$.
\item \label{step:two}
Else if events $L_1$ and $H_2$ occur, then set $\ybf_1=z_d \cdot \ell$ and $\ybf_2=(1-z_d)\cdot \ell$.
\item \label{step:three} 
Else if  events $H_1$ and $H_2$ occur, then set $\ybf_1=\ybf_2=(1-z_d)\cdot \ell$.
\item \label{step:four}
Else if  events $H_1$ and $L_2$ occur, then set $\ybf_1=(1-z_d)\cdot \ell$ and $\ybf_2=z_d\cdot \ell$.
\item \label{step:five}
Else set $\ybf_1=z_d \cdot \ell$ and $\ybf_2=(1-z_d)\cdot \ell$.
\end{enumerate}
\end{tcolorbox}

\begin{lemma}
\label{lem:fixedp-sp}
\fixedp is strategy-proof even when both locations and preferences are private.
\end{lemma}

\begin{proof}
We will prove that there is no deviation that can yield strictly higher utility for any agent $i \in N$.
Fix an arbitrary declaration for all the agents except agent $i$. For every $j \in \{1,2\}$, let
$\ybf_j$, respectively $\ybf'_j$, denote the location where \fixedp places facility $j$ when 
agent $i$ declares truthfully, respectively non-truthfully, his preference for facility $j$.
Let us define $w_{ij}: = |\ybf'_j - x_i| - |\ybf_j - x_i|$. Then, observe that the difference 
$\Delta$ between the utility that agent $i$ gets by reporting truthfully and misreporting, can be written 
as $\Delta = t_{i1}\cdot w_{i1} + t_{i2}\cdot w_{i2}$. Hence, there exists a profitable deviation for
agent $i$ if and only if there is a declaration such that $\Delta < 0$. To prove that such declaration
does not exist, we will use Tables~\ref{tab:mech-cases} and~\ref{tab:sign-of-w-cases}, assuming first agent $i$ is located on $x_i \leq \frac{\ell}{2}$. This is a concise 
representation of all cases that bypasses the repetitive case analysis.
Table~\ref{tab:mech-cases} presents possible preferences of agent $i$ when the mechanism places the 
facilities through Step $k$ and it is interpreted as follows. If $t_{ij}$ at Step $k$ can be 0 or 1, then
we write ``+'' on the corresponding cell of the table; if $t_{ij}$ at Step $k$ can be either 0 or $-1$,
then we write ``-'' on the corresponding cell. For example, if the mechanism places the facilities through 
Step 3, then for agent $i$ it {\em must} hold that $t_{i1} \in \{-1,0\}$ {\em and} $t_{i2} \in \{-1,0\}$.
In Table~\ref{tab:sign-of-w-cases}, the $(k,l)$th cell shows the signs of $w_{i1}$ and $w_{i2}$ when the 
outcome of the mechanism changes from Step $k$ to Step $l$, where Step $k$ corresponds to the outcome when
agent $i$ truthfully declares his preferences and Step $l$ corresponds to the outcome of the mechanism when
agent $i$ lies. 
\hl{Observe that we do {\em not} care how the agent manipulates the outcome. Hence, the agent can misreport his location, his preference, or both.}
So, the $(2,3)$th cell of Table~\ref{tab:sign-of-w-cases} corresponds to the case where 
\fixedp under the true declaration would place the facilities through Step 2, but under the misreport of 
agent $i$ would place them through Step 3. In addition, the signs $(+,0)$ mean that 
$w_{i1} > 0$ and $w_{i2}=0$. So, using this information
alongside the information from the agent's preferences from the third row of Table~\ref{tab:mech-cases}, we 
can deduce that $ \Delta \geq 0$ under this change. If we apply the same reasoning, we will see that 
$\Delta \geq 0$ for all possible cases when $x_i \leq \frac{\ell}{2}$. Hence, there is no profitable deviation 
for agent $i$, when $x_i \leq \frac{\ell}{2}$.

\begin{table}[h!]
    \begin{minipage}{0.35\textwidth}
        \centering
        \begin{tabular}{|c||c|c|}
            \hline
             & $t_{i1}$ & $t_{i2}$\\ \hline \hline 
            Step 1 & +  & + \\ \hline
            Step 2 & + & - \\ \hline
            Step 3 & - & - \\ \hline
            Step 4 & - & + \\ \hline
            Step 5 & + & - \\ \hline
        \end{tabular}
        \caption{Preferences of agent $i$ for every step of \fixedp, when $x_i \leq \frac{\ell}{2}$. \label{tab:mech-cases}}
    \end{minipage}
    \hfill
     \begin{minipage}{0.6\textwidth}
        \centering
        \begin{tabular}{|c||c|c|c|c|c|}
            \hline
            \backslashbox{True}{Lie} & Step 1 & Step 2 & Step 3 & Step 4 & Step 5\\ \hline \hline 
            Step 1 & (0,0)  & (0,+) & (+,+) & (+,0) & (0,+)\\ \hline
            Step 2 & (0,-)  & (0,0) & (+,0) & (+,-) & (0,0)\\ \hline
            Step 3 & (-,-)  & (-,0) & (0,0) & (0,-) & (-,0)\\ \hline
            Step 4 & (-,0)  & (-,+) & (0,+) & (0,0) & (-,+)\\ \hline
            Step 5 & (0,-)  & (0,0) & (+,0) & (+,-) & (0,0)\\ \hline
        \end{tabular}
        \caption{Signs for $(w_{i1}, w_{i2})$ when $x_i \leq \frac{\ell}{2}$. \label{tab:sign-of-w-cases}}
    \end{minipage}
\end{table}

Similarly when $x_i > \frac{\ell}{2}$, we can use Tables~\ref{tab:sign-of-preference-larger} 
and~\ref{tab:sign-of-w-larger} and see again that $\Delta\geq0$ in every case, thus the lemma follows.
Again, for Table~\ref{tab:sign-of-w-larger} \hl{we do not make any assumptions on how the agent manipulated the mechanism, hence we allow him to misreport both his location and his preferences.}

\begin{table}[h!]
    \begin{minipage}{0.35 \textwidth}
        \centering
        \begin{tabular}{|c||c|c|}
            \hline
             & $t_{i1}$ & $t_{i2}$\\ \hline \hline
            Step 1 & -  & - \\ \hline
            Step 2 & - & + \\ \hline
            Step 3 & + & + \\ \hline
            Step 4 & + & - \\ \hline
            Step 5 & - & + \\ \hline
        \end{tabular}
        \caption{Preferences of agent $i$ for every step of \fixedp, when $x_i > \frac{\ell}{2}$. \label{tab:sign-of-preference-larger}}
    \end{minipage}
    \hfill
    \begin{minipage}{0.6\textwidth}
        \centering
        \begin{tabular}{|c||c|c|c|c|c|}
            \hline
            \backslashbox{True}{Lie} & Step 1 & Step 2 & Step 3 & Step 4 & Step 5\\ \hline \hline
            Step 1 & (0,0) & (0,+) & (+,+) & (+,0) & (0,+)\\ \hline
            Step 2 & (0,-)  & (0,0) & (+,0) & (+,-) & (0,0)\\ \hline
            Step 3 & (-,-)  & (-,0) & (0,0) & (0,-) & (-,0)\\ \hline
            Step 4 & (-,0)  & (-,+) & (0,+) & (0,0) & (-,+)\\ \hline
            Step 5 & (0,-)  & (0,0) & (+,0) & (+,-) & (0,0)\\ \hline
        \end{tabular}
        \caption{Signs for $(w_{i1}, w_{i2})$ when $x_i > \frac{\ell}{2}$. \label{tab:sign-of-w-larger}}
    \end{minipage}
\end{table}
\end{proof}

\begin{theorem}
\label{thm:fixedp-appx}
\fixedp is $1 - 2z_d \approx 0.366$-approximate.
\end{theorem}
\begin{proof}
In order to prove our claim, we will focus on the agent that gets the minimum utility
under \fixedp. We will prove that for every possible 
combination of his preferences and his location the agent gets at least $\frac{2z_d}{2-z_d}$ 
the fraction of the utility he would get under an optimal solution. So, let $i$ be an agent
that gets minimum utility under \fixedp. 
Without loss of generality, we will assume that he is located below $\frac{\ell}{2}$. 
Observe that for the preference combinations $(0,1), (1,0), (0,-1), (-1,0)$ the agent
gets utility at least $\ell$, while the maximum utility he can get is trivially bounded by $2\ell$.
Hence, if the agent's preferences are any of these combinations, then under any location for the
facilities the agent gets at least half of his maximum utility and the mechanism is at least 
$\frac{1}{2}$-approximate.
\begin{itemize}
\item $p_i = (1,1)$. Observe that if there exists an agent with preferences $(1,1)$, then 
\fixedp will locate the facilities either through Step~\ref{step:one}, 
or through Step~\ref{step:five}. 
Observe that under any of these steps,  agent $i$ gets utility at least $\ell$, 
while the maximum utility he can get is bounded by $2\ell$. So, the mechanism is 
$\frac{1}{2}$-approximate in any of these steps. 
%
\item $p_i = (1,-1)$. When there exists an agent below $\frac{\ell}{2}$ with preferences $(1,-1)$, 
\fixedp will place the facilities either through Step~\ref{step:two}, 
or through Step~\ref{step:five}. Observe that both steps place the facilities in the same way.
If we check Tables~\ref{tab:main-low} and~\ref{tab:main-high}
we can see that in any case the ratio of the mechanism is greater than $\frac{1}{2}$.
\item  $p_i = (-1,1)$. When there exists an agent below $\frac{\ell}{2}$ with preferences 
$(-1,1)$, then \fixedp will place the facilities either through Step~\ref{step:four}, or 
through Step~\ref{step:five}. In the worst case scenario when Step~\ref{step:four} is used,
agent $i$ is located at $x_i = \frac{\ell}{2} - \epsilon$, for some $\epsilon > 0$. In this
case his utility is $u_i = \ell + 2 \epsilon$, so the ratio of the mechanism is greater than 
$\frac{1}{2}$. 

The utility of $i$ from \fixedp in this step is $u_i = 2z_d \ell$, when $x_i \leq z_d$ and $u_i = 2x_i$, when $z_d < x_i \leq \frac{\ell}{2}$. If Step~\ref{step:five} is chosen, there should exist an agent $i_0$ with $x_{i_0} \geq \frac{\ell}{2}$ and $t_{i_02} = 1$ or with $x_{i_0} < \frac{\ell}{2}$ and $t_{i_02} = -1$. 
The utility of the optimal solution for agent $i$, gets maximized when agent $i_0$ is located in $x_{i_0} = \frac{\ell}{2}$ with preferences $t_{i_0} = (0,1)$. In this case the optimal places $f_1$ on $\ybf_1 = \ell$ and the $f_2$ on $\ybf_2 = \frac{\ell}{4}$. The utility of $i$ under opt is then $u_i = \frac{7 \ell}{4}$, and the approximation in this step is 
\begin{equation}\label{approximation_ratio1}
\frac{8z_d}{7 \ell}
\end{equation}

%
\item  $p_i = (-1,-1)$. When there exists an agent below $\frac{\ell}{2}$ with preferences $(-1,-1)$, then 
\fixedp will place the facilities either through Step~\ref{step:three}, or 
through Step~\ref{step:five}. When Step~\ref{step:three} is used by the mechanism, agent $i$ 
gets utility $(1-z_d-x_i)\cdot 2\ell$ while the optimal value is trivially bounded by $(1-x_i)\cdot 2\ell$. Hence the approximation guarantee from this step for $x_i \leq \frac{\ell}{2}$ is 
\begin{equation}\label{approximation_ratio2}
\frac{1-z_d-x_i}{1-x_i} \geq 1-2z_d
\end{equation}

When Step~\ref{step:five} is used, the worst case instance for the
mechanism is when agent $i$ is located on $z_d$ and there is another agent on $\frac{\ell}{2}-\epsilon$ with preferences $(1,0)$. The optimal mechanism will place $f_1$ on $\ybf_1 = \frac{3\ell +2z_d-2\epsilon}{4}$ and $f_2$ on $\ybf_2 = z_d$.
Then, the utility of agent $i$ in the optimal solution is $\frac{7\ell+2z_d-2\epsilon}{4}$ while
the utility he gets under \fixedp is $(1-2z_d)\cdot\ell$. Hence, the 
approximation ratio of \fixedp is 

\begin{equation}\label{approximation_ratio3}
\frac{4\ell - 8z_d}{7\ell +2z_d}    
\end{equation}

\end{itemize}
We first observe from \eqref{approximation_ratio2} and \eqref{approximation_ratio3} that $1-2z_d \leq \frac{4-8z_d}{7+2z_d}$. Hence, the value of $z_d$ for which the approximation guarantee is maximized can be found if we equalize \eqref{approximation_ratio1} and \eqref{approximation_ratio2}: $\frac{8z_d}{7\ell} = 1-2z_d \Rightarrow z_d = \frac{7}{22} \approx 0.31$. Then the approximation ratio of the mechanism is $\frac{4}{11} \approx 0.36$.
\end{proof}

Observe that since \fixedp asks for the exact location of every agent, it requires 
arbitrarily large communication; this happens for example when the location $x_i$ of an agent $i$ 
is irrational. However, a closer look shows that this is not necessary. An interesting question is whether there exists
a deterministic mechanism that achieves better approximation when every agent communicates $O(1)$ bits.

\begin{theorem}
\label{thm:fixedp-cc}
\hl{The communication complexity of} \fixedp \hl{is 5 bits per agent.}
\end{theorem}
\begin{proof}
\hl{Observe that} \fixedp \hl{can compute its output by only computing which events from $L_1, L_2, H_1, H_2$ occur. This can be done only with a bitstring of 5 bits per agent. So, agent $i$ will send the bit string $(l,f_{11},f_{12}, f_{21}, f_{22})$, where $l$ will contain information about the location $x_i$; $f_{11}$ and $f_{12}$ will contain information about $t_{i1}$; $f_{21}$ and $f_{22}$ will contain information about $t_{i2}$. The agent will send the bits to the mechanism under the following rules.}
\begin{itemize}
    \item \hl{If $x_i \leq \frac{\ell}{2}$, then $l=0$; else $l=1$.}
    \item \hl{If for some $j \in \{1,2\}$ it holds $t_{ij} = 0$, then $f_{j1} = f_{j2} = 0$.}
    \item \hl{If for some $j \in \{1,2\}$ it holds $t_{ij} = 1$, then $f_{j1} = 0$ and $f_{j2} = 1$.}
    \item \hl{If for some $j \in \{1,2\}$ it holds $t_{ij} = -1$, then $f_{j1} = 1$ and $f_{j2} = 1$.}
\end{itemize}
\hl{It is not hard to see that this information suffices for the mechanism to correctly compute which events occur and thus output the correct locations for the facilities.}
\end{proof}

\section{Randomized mechanisms}
In this section, we propose two randomized mechanisms, \rand and \randp that 
achieve constant approximation ratio and are universally strategy-proof and strategy-proof in expectation, respectively, even when both locations and preferences are private information. \rand requires zero communication and
\randp can be implemented using five bits per agent.

\begin{definition}[\rand mechanism]
\rand sets $\ybf_1=\ybf_2=0$ with probability $\frac{1}{2}$ and $\ybf_1=\ybf_2=\ell$ 
with probability $\frac{1}{2}$.
\end{definition}

\begin{theorem}
\label{thm:rand}
\rand is universally strategy-proof and achieves $\frac{1}{2}$ approximation.
\end{theorem}
\begin{proof}
Firstly, it is easy to see that the mechanism is universally strategy-proof
since in each case, the mechanism chooses a fixed location, which is 
strategy-proof.
We will prove that every agent gets utility at least $\frac{\ell}{2}$ in 
expectation from every facility. 
Suppose that agent $i \in N$ is located on $x_i$ and has preferences \ti.
Let us study the expected utility that the agent gets from facility $j$.
If $t_{ij}=1$, then the agent's utility is $\ell - x_i$ when $\ybf_j=0$ and 
$x_i$ when $\ybf_j=\ell$. 
If $t_{ij}=-1$, then the agent gets utility $x_i$ if $\ybf_j=0$ and $\ell-x_i$
if $\ybf_j=\ell$. If $t_{ij}=0$, then the agent gets utility $\ell$ irrespectively
from $\ybf_j$. As a result, the agent gets utility at least 
$\frac{\ell}{2}$ in expectation from each facility. So in total the agent in expectation 
gets utility at least $\ell$. Since, the maximum utility is trivially bounded by $2 \ell$, 
the theorem follows.
\end{proof}

Although \rand seems naive, it achieves the best approximation so far, using zero
communication as well.

\begin{theorem}
    \rand is the optimal mechanism when no communication is allowed.
\end{theorem}

\begin{proof}
For the purpose of contradiction suppose there is a mechanism $M$ achieving an approximation strictly higher than $ \frac{1}{2}$. Let $p(y_1,y_2)$ be the joint probability distribution of the facilities $\ybf_1$ and $\ybf_2$.

Consider now an instance $I_1$ with one agent $a_1$ located at 0 having preferences $(1,1)$. His utility is then $u_1=\int_0^1 \int_0^1 p(\ybf_1, \ybf_2)\cdot (1-\ybf_1) + p(\ybf_1,\ybf_2)\cdot (1-\ybf_2)d\ybf_1d\ybf_2= 2 \int_0^1 \int_0^1 p(\ybf_1, \ybf_2)d\ybf_1d\ybf_2 - \int_0^1 \int_0^1 p(\ybf_1, \ybf_2)(\ybf_1+\ybf_2)d\ybf_1d\ybf_2=2-\omega$ where $\omega=\int_0^1 \int_0^1 p(\ybf_1, \ybf_2)(\ybf_1+\ybf_2)d\ybf_1d\ybf_2$. In $I_1$ the optimal solution places both facilities in 0 resulting in a utility of $u_1^*=2$. The approximation ratio of $M$ in $I_1$ is then 

    \begin{equation}
    \label{I1_approx}
            1-\frac{1}{2}\cdot \omega
    \end{equation}

Similarly consider another instance $I_2$ with one agent $a_2$ located at 0 having preferences $(-1,-1)$. The utility of $a_2$ is then $u_2=\int_0^1 \int_0^1 p(\ybf_1, \ybf_2)(\ybf_1+\ybf_2)d\ybf_1d\ybf_2=\omega$. In $I_2$ the optimal solution places both facilities in 1 resulting in a utility of $u_2^*=2$. Thus the approximation ratio of $M$ in $I_2$ is 
    \begin{equation}
    \label{I2_approx}
            \frac{1}{2}\cdot \omega
    \end{equation}
    Combining \eqref{I1_approx} and \eqref{I2_approx} we derive that $\omega=1$ and thus the approximation of $M$ is $\frac{1}{2}$, a contradiction.

%
%
\end{proof}

We should note that \rand can be extended for 
$k$-facility games, for any $k$, and achieve $\frac{1}{2}$ approximation.
Furthermore, we use the intuition obtained from it to construct \randp. The first four steps of \randp are the same as in \fixedp,
so again we will use the events $L_j$ and $H_j$ introduced in the previous section.

\begin{tcolorbox}[title= \randp mechanism]
\textbf{Input:} Locations $x_1, \ldots, x_n$ and preferences $p_1, \ldots, p_n$.\\
\textbf{Output:} Locations $y_1$ and $y_2$.\\
Set $z_r= \frac{13-\sqrt{161}}{8}$
\begin{enumerate}
\item \label{step:one} 
If events $L_1$ and $L_2$ occur, then set $y_1=y_2=z_r \cdot \ell$.
\item \label{step:two}
Else if events $L_1$ and $H_2$ occur, then set $y_1=z_r \cdot \ell$ and $y_2=(1-z_r) \cdot \ell$.
\item \label{step:three} 
Else if events $H_1$ and $H_2$ occur, then set $y_1=y_2=(1-z_r) \cdot \ell$.
\item \label{step:four}
Else if  events $H_1$ and $L_2$ occur, then set $y_1=(1-z_r) \cdot \ell$ and $y_2=z_r \ell$.
\item \label{step:five}
Else with probability $\frac{1}{2}$ set $y_1=y_2=z_r \cdot \ell$ and with probability $\frac{1}{2}$ set $y_1=y_2=(1-z_r) \cdot \ell$.
\end{enumerate}
\end{tcolorbox}

\begin{lemma}
\label{lem:randp-sp}
\randp is strategy-proof in expectation.
\end{lemma}

\begin{proof}
Steps 1-4  of \randp are similar to the steps of \fixed, so any deviation of agent $i$ that changes the outcome between Steps 1-4 will not result in a better outcome for the agent, as it is shown in Lemma \ref{lem:fixedp-sp}. Therefore we only need to examine any deviation to, or from, Step 5. 

\noindent
\textbf{Deviation to Step 5:} Similarly to Lemma~\ref{lem:fixedp-sp}, let us denote by $\Delta = t_{i1}\cdot w_{i1} + t_{i2}\cdot w_{i2}$ the difference between the expected utility of agent $i$ when he reports truthfully and when misreporting and \randp implements Step 5; here $w_{ij}: = \frac{1}{2}|z_r \ell - x_i| + \frac{1}{2}|(1 - z_r) \ell - x_i| - |\ybf_j - x_i|$. 
Tables \ref{tab:sign-of-preference-smaller-rand-a} and \ref{tab:sign-of-preference-larger-rand-a} present the signs of the preferences of agent $i$ when $x_i \leq \frac{\ell}{2}$ and $x_i > \frac{\ell}{2}$ respectively;
recall ``+'' corresponds to preferences in $\{0,1\}$ and ``-'' to preferences in $\{-1,0\}$. Each cell of the tables \ref{tab:sign-of-w-smaller-rand-a} and \ref{tab:sign-of-w-larger-rand-a} (for $x_i \leq \frac{\ell}{2}$ and $x_i > \frac{\ell}{2}$ respectively) presents the signs of $(w_{i1}, w_{i2})$ when agent $i$ misreports so that Step 5 is followed by \randp. It can be easily verified that for all possible combinations we have that
$\Delta \geq 0$. Thus, any misrepresentation of the preferences of agent $i$ that changes the outcome of \randp
from Step 1-4 to Step 5 does not increase the utility of the agent.

\begin{table}[h!]
    \begin{minipage}{0.45 \textwidth}
        \centering
        \begin{tabular}{|c||c|c|}
            \hline
             & $t_{i1}$ & $t_{i2}$\\ \hline \hline
            Step 1 & +  & + \\ \hline
            Step 2 & + & - \\ \hline
            Step 3 & - & - \\ \hline
            Step 4 & - & + \\ \hline
        \end{tabular}
        \caption{Preferences of agent $i$ for Steps 1-4 of \randp, when $x_i \leq \frac{\ell}{2}$. \label{tab:sign-of-preference-smaller-rand-a}}
    \end{minipage}
    \hfill
    \begin{minipage}{0.5\textwidth}
        \centering
        \begin{tabular}{|c||c|}
            \hline
            \backslashbox{True}{Lie} & Step 5 \\ \hline \hline
            Step 1 & (+,+) \\ \hline
            Step 2 & (+,-) \\ \hline
            Step 3 & (-,-) \\ \hline
            Step 4 & (-,+) \\ \hline
        \end{tabular}
        \caption{Signs for $(w_{i1}, w_{i2})$ when $x_i \leq \frac{\ell}{2}$. \label{tab:sign-of-w-smaller-rand-a}}
    \end{minipage}
\end{table}

\begin{table}[h!]
    \begin{minipage}{0.45 \textwidth}
        \centering
        \begin{tabular}{|c||c|c|}
            \hline
             & $t_{i1}$ & $t_{i2}$\\ \hline \hline
            Step 1 & -  & - \\ \hline
            Step 2 & - & + \\ \hline
            Step 3 & + & + \\ \hline
            Step 4 & + & - \\ \hline
        \end{tabular}
        \caption{Preferences of agent $i$ for Step 1-4 of \randp, when $x_i > \frac{\ell}{2}$. \label{tab:sign-of-preference-larger-rand-a}}
    \end{minipage}
    \hfill
    \begin{minipage}{0.45\textwidth}
        \centering
        \begin{tabular}{|c||c|}
            \hline
            \backslashbox{True}{Lie} & Step 5 \\ \hline \hline
            Step 1 & (-,-) \\ \hline
            Step 2 & (-,+) \\ \hline
            Step 3 & (+,+) \\ \hline
            Step 4 & (+,-) \\ \hline
        \end{tabular}
        \caption{Signs for $(w_{i1}, w_{i2})$ when $x_i > \frac{\ell}{2}$. \label{tab:sign-of-w-larger-rand-a}}
    \end{minipage}
\end{table}

\noindent
\textbf{Deviation from Step 5:} In order for the outcome of the mechanism to change, the preference of $i$ must change sign from ``+'' to ``-'' or from ``-'' to ``+''. Similarly as above, 
$\Delta = t_{i1}\cdot w_{i1} + t_{i2}\cdot w_{i2}$ denotes the difference between the expected utility of agent $i$ when he reports truthfully and Step 5 is employed, and when misreporting. Now, we have that 
$w_{ij}: = |\ybf_j - x_i| - \frac{1}{2}|z_r \ell - x_i| - \frac{1}{2}|(1 - z_r) \ell - x_i|$. Each cell of column $c$ of the tables \ref{tab:sign-of-t-smaller-rand-b} and \ref{tab:sign-of-t-larger-rand-b} presents the only possible signs of $(t_{i1}, t_{i2})$ of agent $i$ in Step 5 when he reports his true preference, that can result in the step of column $c$ if he misreports. As an example, consider the case where $t_{i1} \leq 0$ when $x_i \leq \frac{\ell}{2}$. If $i$ changes his declaration to $t_{i1}' \geq 0$ Step 1 cannot be followed. 
Again, it can be easily verified that $\Delta \geq 0$ in every case, hence there is no profitable deviation for
agent $i$.

\begin{table}[h!]
    \begin{minipage}{0.45 \textwidth}
        \centering
        \begin{tabular}{|c||c|c|c|c|}
            \hline
            \backslashbox{True}{Lie} & Step 1 & Step 2 & Step 3 & Step 4 \\ \hline \hline
            Step 5 & (-,-) & (-,+) & (+,+) & (+,-) \\ \hline
        \end{tabular}
        \caption{Signs of $(t_{i1}, t_{i2})$ when $x_i \leq \frac{\ell}{2}$. \label{tab:sign-of-t-smaller-rand-b}}
    \end{minipage}
    \hfill
    \begin{minipage}{0.45\textwidth}
        \centering
        \begin{tabular}{|c||c|c|c|c|}
            \hline
            \backslashbox{True}{Lie} & Step 1 & Step 2 & Step 3 & Step 4 \\ \hline \hline
            Step 5 & (-,-) & (-,+) & (+,+) & (+,-) \\ \hline
        \end{tabular}
        \caption{Signs of $(w_{i1}, w_{i2})$ when $x_i \leq \frac{\ell}{2}$. \label{tab:sign-of-w-smaller-rand-b}}
    \end{minipage}
\end{table}

\begin{table}[h!]
    \begin{minipage}{0.45 \textwidth}
        \centering
         \begin{tabular}{|c||c|c|c|c|}
            \hline
            \backslashbox{True}{Lie} & Step 1 & Step 2 & Step 3 & Step 4 \\ \hline \hline
            Step 5 & (+,+) & (+,-) & (-,-) & (-,+) \\ \hline
        \end{tabular}
        \caption{Signs of $(t_{i1}, t_{i2})$ when $x_i > \frac{\ell}{2}$. \label{tab:sign-of-t-larger-rand-b}}
    \end{minipage}
    \hfill
    \begin{minipage}{0.45 \textwidth}
        \centering
        \begin{tabular}{|c||c|c|c|c|}
            \hline
            \backslashbox{True}{Lie} & Step 1 & Step 2 & Step 3 & Step 4 \\ \hline \hline
            Step 5 & (+,+) & (+,-) & (-,-) & (-,+) \\ \hline
        \end{tabular}
        \caption{Signs of $(w_{i1}, w_{i2})$ when $x_i > \frac{\ell}{2}$. \label{tab:sign-of-w-larger-rand-b}}
    \end{minipage}
\end{table}
\end{proof}

\begin{theorem}
\label{thm:randp-apx}
\randp is $(\frac{1}{2}+z_r) \simeq 0.538$-approximate.
\end{theorem}

\begin{proof}
To prove our claim, we will focus on the agent that gets the minimum utility
under \randp. We will prove that for every possible 
combination of his preferences and his location the agent gets a fraction of
$\frac{1}{2}+z_r$ of the utility he would get under an optimal solution. So, let $i$ be an agent
that gets minimum utility under \randp. Without loss of
generality, we will assume that he is located below $\frac{\ell}{2}$. 
\begin{itemize}
\item $p_i = (1,1)$. If there exists an agent below $\frac{\ell}{2}$ with 
preferences $(1,1)$, then \randp will place the facilities through  
Step~\ref{step:one}, or through Step~\ref{step:five}. 
We will consider each case separately.
If the facilities are placed due to Step~\ref{step:one}, then the utility of the 
agent is at least $(1+2z_r)\ell$, while the maximum utility the agent can get is $2\ell$.
Hence, the approximation guarantee of the mechanism, in this case, is $\frac{1}{2} +z_r$.
For the case where the facilities are placed due to Step~\ref{step:five}, we have to consider
the following subcases. Firstly, if $x_i\geq z_r$, then the expected utility of the agent 
is $(1+2z_r)\ell$ while the optimum is bounded by $2\ell$, hence the mechanism achieves 
$\frac{1}{2}+z_r$ approximation.
If $x_i<z_r$, we have to further consider two cases depending on the reason Step~\ref{step:five}
was triggered. The first one is that there exists an agent $i'$ with $x_{i'}<\frac{\ell}{2}$ 
that has preference $-1$ for one of the two facilities. Then, the optimal value for the objective
is upper bounded by $\frac{3\ell}{2}$. Hence, since we assumed that agent $i$ has the minimum utility
under \randp, we get that it achieves $\frac{2}{3}$ approximation in this case. 
The second subcase is when there exists an agent $i'$ with $x_{i'}>\frac{\ell}{2}$ that has preference $1$ 
for one of the two facilities. Then, the optimum is again upper bounded by $\frac{3\ell}{2}$ and the mechanism achieves the claimed approximation ratio.

\item $p_i = (1,0)$. Observe that the analysis for the case $p_i=(0,1)$ is symmetric, hence it
will be omitted. When there exists an agent below $\frac{\ell}{2}$ with preferences $(0,1)$, then 
\randp will place the facilities either through Step~\ref{step:one}, or through 
Step~\ref{step:two}, or through 
Step~\ref{step:five}. As in the previous case, it is not hard to see that under any of these steps the
utility of agent $i$ is at least $(\frac{3\ell}{2}+z_r)\cdot \ell$, while the maximum utility he can get is bounded 
by $2\ell$. So, the approximation guarantee follows.
\item $p_i = (1,-1)$. Observe that the analysis for the case $p_i=(-1,1)$ is symmetric, hence it
will be omitted. When there exists an agent below $\frac{\ell}{2}$ with preferences $(-1,1)$, 
then \randp will locate the facilities either through 
Step~\ref{step:two}, or through Step~\ref{step:five}. When the mechanism locates the facilities
through Step~\ref{step:two}, then the worst-case instance occurs when there is only one agent $i$ on $\frac{\ell}{2}$. Then,
the agent gets utility $\ell$, while an optimal solution locates the first facility on $\frac{\ell}{2}$
and the second facility on 0 yielding utility $\frac{3\ell}{2}$. Thus, the mechanism is 
$\frac{2}{3}$-approximate. So, for the chosen value of $z_r$, the mechanism is $(\frac{1}{2}+z_r)$-approximate.
If, on the other hand, the mechanism locates the facilities through Step~\ref{step:five}, then the expected utility of agent $i$ is $\ell$ irrespectively of his location $x_i$. In order to construct a worst-case
instance, it suffices to consider instances with only two agents, since more agents can only restrict 
more the set of optimal solutions, which implies that the optimal value can only decrease. The ``loosest''
constraint to the optimal that triggers Step~\ref{step:five} too, is when there is an agent on 
$\frac{\ell}{2} + \epsilon$ for an arbitrarily small positive $\epsilon$ with preferences $(1,0)$. 
Then, the worst-case instance for the mechanism, in terms of approximation guarantee, is when 
$x_i=0$ and the optimal utility the agent $i$ gets is bounded by $\frac{7\ell}{4}$; the first facility is 
located at $\frac{\ell}{4}$ and the second one at $\ell$. Hence, the mechanism is 
$\frac{4}{7}$-approximate. So, for the chosen value of $z_r$, the mechanism is 
$(\frac{1}{2}+z_r)$-approximate.
\item  $p_i = (-1,0)$. Observe that the analysis for the case $p_i=(0,-1)$ is symmetric, hence it
will be omitted. When there exists an agent below $\frac{\ell}{2}$ with preferences $(-1,0)$, then 
\randp will place the facilities either through Step~\ref{step:three}, or 
through Step~\ref{step:four}, or through Step~\ref{step:five}. When Step~\ref{step:three}, or 
Step~\ref{step:four}, is used,
agent $i$ gets utility $2\ell-z_r \ell -x_i$, while the optimal utility is bounded by $2\ell-x_i$ by locating both
facilities on $\ell$. So, since $x_i \leq \frac{\ell}{2}$, the approximation of \randp 
 is at least $1-\frac{z_r}{2} > \frac{1}{2}+ z_r$, for the chosen value of $z_r$.
If Step~\ref{step:five} is used, then agent $i$ gets utility at least $(\frac{3}{2}-z_r)\cdot \ell$. Furthermore,
we can trivially bound the optimal utility by $2\ell$, so, for the chosen value of $z_r$ we get that
\randp is $(\frac{1}{2}+z_r)$-approximate.
\item  $p_i = (-1,-1)$. When there exists an agent below $\frac{\ell}{2}$ with preferences $(-1,-1)$, then 
\randp will locate the facilities either through Step~\ref{step:three}, or 
through Step~\ref{step:five}. When Step~\ref{step:three} is used by the mechanism, then agent $i$ 
gets utility $(1-z_r-x_i)\cdot 2\ell$ while the optimal value is trivially bounded by $(1-x_i)\cdot 2\ell$.
Hence, since $x_i \leq \frac{\ell}{2}$, the approximation guarantee of the mechanism is bounded by 
$1-\frac{z_r}{2} > \frac{1}{2}+z_r$.
If Step~\ref{step:five} is used, then using similar arguments as in the case where $p_i=(1,-1)$, we
can construct the worst case instance for the mechanism by locating an agent with preferences $(1,0)$
on $\frac{\ell}{2}-\epsilon$ and by setting $x_i=z_r \ell$. Then, under Step~\ref{step:five} agent $i$ gets
utility $(1-2z_r)\cdot \ell$ and his utility under the optimal location for the facilities is bounded by 
$(\frac{7}{4}-z_r)\cdot\ell$; the first facility is located on $(\frac{3}{4}+z_r)\cdot \ell$ and the second 
one on $\ell$. Hence, \randp is $(1-2z_r)/(\frac{7}{4}-z_r)$-approximate.
It is not hard to verify that for the chosen value of $z_r$ we get that $(1-2z_r)/(\frac{7}{4}-z_r)= 
\frac{1}{2}+z_r$.
\end{itemize}
\end{proof}

Using exactly the same arguments as in Theorem \ref{thm:fixedp-cc}we can get that \randp can be implemented in a communication-efficient way where each agent sends only five bits to the planner.

\begin{theorem}
\label{thm:rand-cc}
The communication complexity of \randp is 5 bits per agent.
\end{theorem}



    

\section{Two-preference instances}
In this section, we study $k$-facility games where all the agents have preferences in $\{0,1\}^k$, 
$\{1,-1\}^k$, or in $\{0,-1\}^k$, which we call two-preference instances. 
The non-existence of optimal deterministic strategy-proof mechanisms can be 
extended even on two-preference instances with three agents. 

\begin{theorem}
\label{thm:utility}
For any $k \geq 2$, there is no optimal deterministic strategy-proof mechanism 
for $k$-facility games even on two-preference instances with three 
agents and known locations.
\end{theorem}

The proof of the theorem follows from the instances of Figures \ref{fig2}, \ref{fig3}, \ref{fig4}. As in Theorem \ref{thm:2fub}, white circles correspond to agents and black circles to the optimal locations.


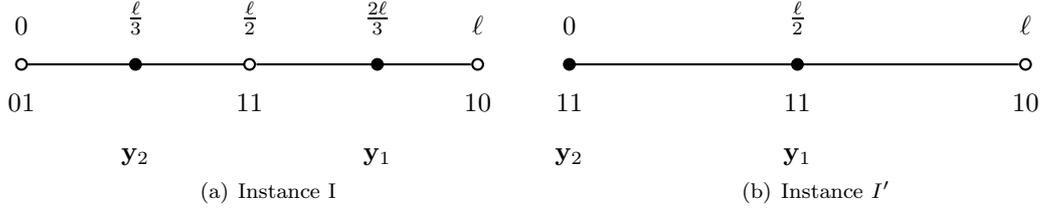
\begin{figure}[h!]
\begin{center}
\subfigure[Instance I]{
\begin{tikzpicture}[thick, scale=0.6]
  \tikzstyle{every node}==[fill=white,minimum size=4pt,inner sep=0pt]

\draw (-4,-0.5) node(v)[label=below:$0 1$]{};
\draw (-4,0.5) node(v1)[label=above:$0$]{};
\draw (-1.5,0.5) node(v2)[label=above:$\frac{\ell}{3}$]{};
\draw (1,-0.5) node(v2)[label=below:$1 1$]{};
\draw (6,0) node(v3)[draw, fill=white, circle]{};
\draw (-4,0) node(v7)[draw, fill=white, circle]{};
\draw (1,0) node(v7)[draw, fill=white, circle]{};
\draw (-1.5,0) node(v8)[draw, fill=black, circle]{};
\draw (3.8,0) node(v9)[draw, fill=black, circle]{};
\draw (1,0.5) node(v2)[label=above:$\frac{\ell}{2}$]{};
\draw (-1.5,-1.7) node(u)[label=below:$\ybf_2$]{};
\draw (3.8,-1.7) node(u1)[label=below:$\ybf_1$]{};
\draw (6,-0.5) node(l1)[label=below :$1 0$]{};
\draw (6,0.5) node(u1)[label=above :$\ell$]{};
\draw (3.8,0.5) node(v6)[label=above:$\frac{2 \ell}{3}$]{};
\draw (-3.87,0) -- (0.87,0);
\draw (1.15,0) -- (5.85,0);
\node[] at (7,-2) {};
\end{tikzpicture}
}
\subfigure[Instance $I'$]{
\begin{tikzpicture}[thick, scale=0.6]
  \tikzstyle{every node}==[fill=white,minimum size=4pt,inner sep=0pt]

\draw (-4,-0.5) node(v)[label=below:$1 1$]{};
\draw (-4,-1.7) node(v)[label=below:$\ybf_2$]{};
\draw (-4,0.5) node(v1)[label=above:$0$]{};
\draw (1,-0.5) node(v2)[label=below:$1 1$]{};
\draw (1,-1.7) node(v)[label=below:$\ybf_1$]{};
\draw (6,0) node(v3)[draw, fill=white, circle]{};
\draw (-4,0) node(v7)[draw, fill=black, circle]{};
\draw (1,0) node(v7)[draw, fill=black, circle]{};
\draw (1,0.5) node(v2)[label=above:$\frac{\ell}{2}$]{};
\draw (6,-0.5) node(l1)[label=below :$1 0$]{};
\draw (6,0.5) node(u1)[label=above :$\ell$]{};
\draw (-4,0) -- (5.85,0);
\node[] at (7,-2) {};
\end{tikzpicture}
}
\caption{Example for preferences in $\{0,1\}^2$. The agent located on $0$ in the 
instance $I$ can declare preferences $(1,1)$ and increase his utility by moving 
the facility $f_2$ closer to 0.}
\label{fig2}
\end{center}
\end{figure}
\begin{figure}[h!]
\begin{center}
\subfigure[Instance I]{
\begin{tikzpicture}[thick, scale=0.6]
  \tikzstyle{every node}==[fill=white,minimum size=4pt,inner sep=0pt]

\draw (-4.2,-0.2) node(v)[label=below:-1 1]{};
\draw (-4,0.2) node(v1)[label=above:$0$]{};
\draw (5,-0.2) node(v2)[label=below:1 1]{};
\draw (5,0) node(v3)[draw, fill=white, circle]{};
\draw (-4,0) node(v7)[draw, fill=white, circle]{};
\draw (1,0) node(v7)[draw, fill=black, circle]{};
\draw (6,0) node(v3)[draw, fill=black, circle]{};
\draw (1,0.5) node(v2)[label=above:$\frac{\ell}{2}$]{};
\draw (5,0.5) node(v4)[label=above:$\ell$-$\epsilon$]{};
\draw (1,-1.7) node(u)[label=below:$\ybf_2$]{};
\draw (6,-1.7) node(l1)[label=below :$\ybf_1$]{};
\draw (6,0.5) node(u1)[label=above :$\ell$]{};
\draw (-3.85,0) -- (4.85,0);
\draw (5.15,0) -- (5.85,0);
\node[] at (7,-2) {};
\end{tikzpicture}
}
\subfigure[Instance $I'$]{
\begin{tikzpicture}[thick, scale=0.6]
  \tikzstyle{every node}==[fill=white,minimum size=4pt,inner sep=0pt]

\draw (-4.2,-0.2) node(v)[label=below:-1 1]{};
\draw (-4,0.2) node(v1)[label=above:$0$]{};
\draw (4.8,-0.2) node(v2)[label=below:-1 1]{};
\draw (5,0) node(v3)[draw, fill=black, circle]{};
\draw (-4,0) node(v7)[draw, fill=white, circle]{};
\draw (6,0) node(v3)[draw, fill=black, circle]{};
\draw (5,0.5) node(v4)[label=above:$\ell$-$\epsilon$]{};
\draw (4.8,-1.7) node(u)[label=below:$\ybf_2$]{};
\draw (6.2,-1.7) node(l1)[label=below :$\ybf_1$]{};
\draw (6,0.5) node(u1)[label=above :$\ell$]{};
\draw (-3.85,0) -- (5.85,0);
\node[] at (7,-2) {};
\end{tikzpicture}
}
\caption{Example for preferences in $\{-1,1\}^2$. The agent located on 
$\ell-\epsilon$ in the instance $I$ can declare preferences $(-1,1)$ and 
increase his utility by moving the facility $f_2$ closer to $\ell-\epsilon$.
}
\label{fig3}
\end{center}
\end{figure}


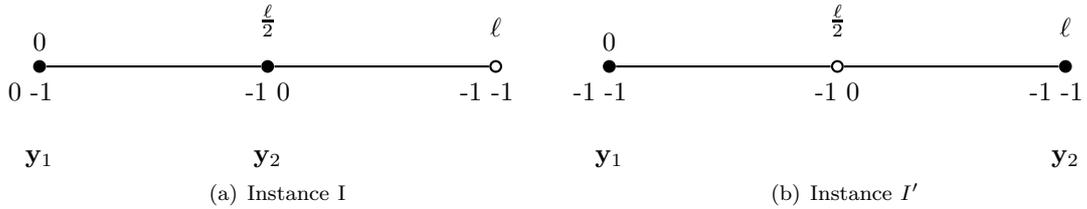
\begin{figure}[h!]
\begin{center}
\subfigure[Instance I]{
\begin{tikzpicture}[thick, scale=0.6]
  \tikzstyle{every node}==[fill=white,minimum size=4pt,inner sep=0pt]

\draw (-4.2,-0.2) node(v)[label=below:0 -1]{};
\draw (1,-0.2) node(v)[label=below:-1 0]{};
\draw (-4,0.2) node(v1)[label=above:$0$]{};
\draw (5.8,-0.2) node(v2)[label=below:-1 -1]{};
\draw (-4,0) node(v7)[draw, fill=black, circle]{};
\draw (1,0) node(v7)[draw, fill=black, circle]{};
\draw (6,0) node(v3)[draw, fill=white, circle]{};
\draw (1,0.5) node(v2)[label=above:$\frac{\ell}{2}$]{};
\draw (1,-1.7) node(u)[label=below:$\ybf_2$]{};
\draw (-4,-1.7) node(l1)[label=below :$\ybf_1$]{};
\draw (6,0.5) node(u1)[label=above :$\ell$]{};
\draw (-3.85,0) -- (0.85,0);
\draw (1.15,0) -- (5.85,0);
\node[] at (7,-2) {};
\end{tikzpicture}
}
\subfigure[Instance $I'$]{
\begin{tikzpicture}[thick, scale=0.6]
  \tikzstyle{every node}==[fill=white,minimum size=4pt,inner sep=0pt]

\draw (-4.2,-0.2) node(v)[label=below:-1 -1]{};
\draw (1,-0.2) node(v)[label=below:-1 0]{};
\draw (-4,0.2) node(v1)[label=above:$0$]{};
\draw (5.8,-0.2) node(v2)[label=below:-1 -1]{};
\draw (-4,0) node(v7)[draw, fill=black, circle]{};
\draw (1,0) node(v7)[draw, fill=white, circle]{};
\draw (6,0) node(v3)[draw, fill=black, circle]{};
\draw (1,0.5) node(v2)[label=above:$\frac{\ell}{2}$]{};
\draw (6,-1.7) node(u)[label=below:$\ybf_2$]{};
\draw (-4,-1.7) node(l1)[label=below :$\ybf_1$]{};
\draw (6,0.5) node(u1)[label=above :$\ell$]{};
\draw (-3.85,0) -- (0.85,0);
\draw (1.15,0) -- (5.85,0);
\node[] at (7,-2) {};
\end{tikzpicture}
}
\caption{Example for preferences in $\{-1,0\}^2$. The agent located on $0$ in 
the instance $I$ can declare preferences $(-1,-1)$ and increase his utility by 
moving the facility $f_2$ away from 0. Observe that for the Instance $I'$ there 
are two optimal solutions ($\ybf_1=0, \ybf_2=\ell$ and $\ybf_1=\ell, \ybf_2=0$). 
However, this does not affect the correctness of our example assuming that the 
mechanism chooses a solution \emph{deterministically}.}
\label{fig4}
\end{center}
\end{figure}

We now show how we can modify \fixed by changing the value of $z_f$ and achieve 
better approximation guarantees. We denote the mechanisms as \fixedzo, for preferences in $\{0,1\}^k$,
and \fixedzm, for preferences in $\{-1,0\}^k$.
Furthermore, for $k=2$ we derive a new deterministic mechanism
termed $OPT^2$, for the case where all agents have preferences in $\{0,1\}^2$ and their locations
are known. 

\begin{definition}
\fixedzo sets $\ybf_1=\ldots=\ybf_k=\frac{\ell}{2}$.
\end{definition}

\begin{theorem}
\label{thm:fzo}
\fixedzo is $\frac{1}{2}$-approximate.
\end{theorem}
\begin{proof}
Observe that for every agent $i$ and any facility $j$ it holds that 
$u_{ij}(x_i,t_{ij},\ybf_j) \geq \ell - |x_i-\frac{\ell}{2}| \geq \frac{\ell}{2}$. Hence, 
$u_i(x_i,t_i,\ybf) \geq  \frac{k\cdot\ell}{2}$. Observe, however, that 
$\max_\ybf u_i(x_i,t_i,\ybf) \leq k\cdot \ell$. Hence, agent $i$ under \ybf gets 
at least half of his maximum utility.
\end{proof}

\begin{definition}
\fixedzm sets $\ybf_1=\ldots=\ybf_{\lceil \frac{k}{2}\rceil}=0$ and 
$\ybf_{\lfloor \frac{k}{2}\rfloor}=\ldots = \ybf_k = \ell$.
\end{definition}

\begin{theorem}
\label{thm:fzm}
\fixedzm is $\frac{\lfloor \frac{k}{2}\rfloor}{k}$-approximate.
\end{theorem}
\begin{proof}
Observe that since $t_i \in \{0,-1\}^k$ it holds that 
$u_i(x_i,t_i,\ybf) = \sum_j \lceil \frac{k}{2}\rceil \cdot x_i + \lfloor \frac{k}{2}\rfloor 
\cdot (\ell -x_i) \geq  \lfloor \frac{k}{2}\rfloor \cdot \ell$. Observe though that 
$\max_\ybf u_i(x_i,t_i,\ybf) \leq k\cdot \ell$. Hence, \fixedzm is at least 
$\frac{\lfloor \frac{k}{2}\rfloor}{k}$-approximate.
\end{proof}

\begin{definition}
$OPT^2$ places each of the two facilities independently on its optimal location.
\end{definition}

It is not hard to see that $OPT^2$ is strategy-proof. This is because we know that 
when agents' locations are known, the mechanism that places one facility on the 
leftmost optimal location is strategy-proof. Therefore, since the mechanism places each 
facility independently no agent can increase his utility by lying.

\begin{theorem}
\label{thm:opt2}
$OPT^2$ is $\frac{3}{4}$-approximate.
\end{theorem}

\begin{proof}
Before we analyze the approximation guarantee of the mechanism, let us first 
study the locations in which the mechanism places the facilities. Since the
preferences of each agent are in $\{0,1\}^2$, it is not hard to see that 
the optimal location for each facility is the median point between the locations
of the leftmost and the rightmost agents that want to be close to the facility.

Without loss of generality, we can assume that the agent with the minimum utility
under $\opt^2$, denoted by $a_1$, has preferences $(1,1)$. 
If $t_i =(1,0)$, then the agent would have utility at least $\frac{3}{2}\ell$ 
since any other agent who wants to be close to the first facility is located 
in distance at most $\ell$ from $a_1$'s location. 
The maximum utility the agent can get is $2\ell$, so the mechanism is then
$\frac{3}{4}$-approximate. 

Assume that $a_1$ is located on $x \leq \frac{\ell}{2}$. 
Then, without loss of generality, we can assume that he is located on 0 since for
any other location the agent would be closer to the facilities and thus his 
utility would increase.
Then, observe that agent $a_1$, alongside with the rightmost agents, 
will define the locations of the facilities. Observe that if the 
rightmost agent has preferences $(1,1)$, then $\opt^2$ is optimal. So, we can 
assume that the rightmost agent, denoted by $a_{r1}$, has preferences $(0,1)$. 
In the worst case, $a_{r1}$ will be located on $\ell$, since for 
every other location, the utility of agent $a_1$ will be higher. We have 
to consider the two possible preferences for the second rightmost agent with 
preference 1 for the first facility and prove that $\opt^2$ achieves the desired
approximation.
We will use $a_i$ to denote this agent and $x_i$ to denote his location.

Firstly, we consider the case where agent $a_i$ has preferences $(1,1)$ and
$x_i \geq \frac{\ell}{2}$.
The utilities of the agents for the facilities under the locations $(\ybf_1,\ybf_2)$, where 
$\ybf_2 \leq x_i$, are $u_1 = 2 \ell -\ybf_1-\ybf_2$, 
$u_i = 2 \ell -2x_i+\ybf_1+\ybf_2$ and $u_{r_1}=\ell +\ybf_2$. 
$\opt^2$ will place the facilities to $\ybf_1=\frac{x_i}{2}$ and 
$\ybf_2=\frac{\ell}{2}$ and the utility of agent $a_1$ will be 
$u_1=\frac{3 \ell  -x_i}{2}$. Observe that the locations of the facilities that make 
the utilities of these three agents equal provide an upper bound on the utility
that agent $a_1$ gets under the optimal solution, since any other solution would
yield lower utility for at least one of these agents. If we find the locations
of the facilities that equalize the utilities for the agents we get 
$\ybf_1=2x_i-\ell$ and $\ybf_2=\ell-x_i$ and thus the optimal utility for 
agent $a_1$ is bounded by $2\ell-x_i$. 
Hence, $\opt^2$ is $\alpha=\frac{3-x_i}{4-2x_i} \geq \frac{3}{4}$-approximate.

In the case where $x_i<\frac{\ell}{2}$, it is not difficult to see that 
agent $a_1$ gets utility at least $\frac{5}{4}\ell$ under $\opt^2$. Observe that 
under the optimal solution the utility of the agents is bounded by 
$\frac{3}{2}\ell$, since there are no locations for the facilities where both 
$a_1$ and $a_{r1}$ get more than $\frac{3}{2}\ell$. Thus, in this case the 
mechanism is $\frac{5}{6}$-approximate.

If the preferences of $a_i$ are $(1,0)$, then similar analysis can be applied.
\end{proof}

\section{\welfare and \happy}
In this section we show that \fixed, \fixedzo, \fixedzm, and \rand achieve the same approximation guarantees
for \welfare and \happy objectives as \util. All mechanisms remain strategy-proof since they do not 
require any information from the agents. Recall, \welfare is the sum of the utilities of the agents, formally 
$\sum_i u_i(x_i, t_i, \ybf)$ and \happy is $\min_i \frac{u_i(x_i, t_i, \ybf)}{u^*_i(x_i, t_i)}$, where 
$u^*_i(x_i, t_i) = max_\ybf u_i(x_i, t_i,\ybf)$.
\begin{theorem}
\label{thm:wel-hap}
For \welfare and \happy objectives the following hold.
\fixed is $z_f$-approximate. 
\fixedzo is $\frac{1}{2}$-approximate.
\fixedzm is $\frac{\lfloor \frac{k}{2}\rfloor}{k}$-approximate.
\rand is $\frac{1}{2}$-approximate.
\end{theorem}
\begin{proof}
In the proofs of Theorems~\ref{thm:mech2},~\ref{thm:fzo},~\ref{thm:fzm}, and~\ref{thm:randp-apx} 
it is proved that for every agent $i$ it holds that $\frac{u_i(x_i, t_i, \ybf)}{u^*_i(x_i, t_i)} \geq \alpha$, 
where $\alpha$ is the approximation ratio of the corresponding mechanism. Hence, the claim for \happy 
already follows from those proofs since they capture the definition of \happy. For 
\welfare, observe that $OPT_w=\max_\ybf \sum_i u_i(x_i, t_i, \ybf) \leq \sum_i u^*_i(x_i, t_i)$.
So, from the proofs of the aforementioned theorems we get that 
$u_i(x_i, t_i, \ybf) \geq u^*_i(x_i, t_i)\cdot \alpha$
for every $i$. So, if we sum over $i$ we get that 
$\sum_i u_i(x_i, t_i, \ybf) \geq \alpha \cdot \sum_i u^*_i(x_i, t_i) \geq \alpha \cdot OPT_w$
and the theorem follows.
\end{proof}

The observing reader may wonder whether the approximation guarantee of \fixed for \welfare
contradicts the result of~\cite{ZL15}. Recall, \cite{ZL15} proved that there is no deterministic
strategy-proof mechanism for \welfare with approximation ratio better than $\frac{2}{n}$.
However, a closer look will reveal that in order to establish that result, the following assumptions must be made.
Firstly, that every agent wants to be close to the first facility and away from the second facility.
Furthermore, they defined the utility of an agent located on $x_i$ to be 
$u_i(x_i, \ybf)= |x_i-\ybf_1| - |x_i-\ybf_2|$. This different definition of utility is crucial
for deriving those negative results, and this is the reason why our results do not contradict theirs.
\section{Discussion}
In this paper, we studied heterogeneous facility locations on the line segment. To the best of our knowledge, this
is the first systematic study of this model for the \util objective. We derived inapproximability results for
strategy-proof mechanisms for \util even for instances with known locations and two agents. Furthermore, 
we derived strategy-proof mechanisms that achieve constant approximation for \util, some of which also achieve the same guarantee for \welfare and \happy objectives.

All of our mechanisms are simple and can be implemented in a communication-efficient way. More specifically,
every mechanism needs zero or five bits of information from every agent. Communication efficiency is 
crucial for real-life scenarios. Consider the example of the factory and the school discussed in the
introduction. If thousand of citizens live on this street, then our mechanisms require only their preferences 
and whether they live on the western part of the street or on the east one and not their full address saving a huge 
amount of time to the planner. To the best
of our knowledge, this is the first time that communication complexity is studied for facility location problems.
We strongly believe that there is much to be said about facility location mechanisms and communication 
complexity. Firstly, it would be really interesting to understand how limited communication affects the approximation
guarantee of mechanisms. Is there a better randomized mechanism than \rand when no communication is allowed?
Are there better mechanism than \fixedp and \randp when every agent is allowed to communicate $O(1)$ bits? 
Can \fixedp and \randp be extended for $k \geq 3$ facilities?

Another intriguing avenue of research is to use communication complexity to define ``simple''
mechanisms. Recently~\citeauthor{LiOSP}~\cite{LiOSP} defined the \emph{obviously strategy-proof} (OSP)
mechanisms to capture the simplicity of mechanisms. Intuitively, a mechanism is obviously 
strategy-proof if it remains incentive-compatible even when some of the agents are not fully rational.
The formal definition of OSP is quite technical, and thus we decided not to include it in our paper since it
would deviate from its main theme. However, we strongly believe that some of our mechanisms, if not all of 
them, should be \emph{obviously strategy-proof}~\cite{LiOSP}. \fixed and \rand do not use any information
from the agents. In both \fixedp and \randp, if an agent knows the declarations of the rest of the agents, then 
he can verify that he cannot increase his utility by misreporting his type using $O(1)$ space. We believe that
this kind of mechanisms are de facto simple and deserve further studying.

\newpage

\bibliographystyle{plainnat}
\bibliography{references}

\end{document}